\documentclass[acmsmall]{acmart}
\AtBeginDocument{%
  }

\setcopyright{cc}
\setcctype{by}
\acmJournal{PACMMOD}
\acmYear{2026} \acmVolume{4} \acmNumber{2 (PODS)} \acmArticle{111}
\acmMonth{5} \acmDOI{10.1145/3801907}

\usepackage{mathrsfs}
\usepackage{graphicx}
\usepackage{tabularx}
\usepackage{tikz}
\usepackage{commath}
\usepackage{braket}
\usepackage{wasysym}
\usepackage{booktabs}
\usepackage{hyperref}
\usepackage{amsmath}
\usepackage{amsthm}
\usepackage{lipsum}
\usepackage{url}
\usepackage{bm}
\usepackage{mathtools}
\usepackage[nameinlink]{cleveref}

\usepackage{algorithm}
\usepackage{algorithmic}

\renewcommand{\phi}{\varphi}

\hypersetup{
    colorlinks = true,
    linkcolor={ref-color},
    urlcolor={ref-color},
    citecolor={blue!80!black}
}

\newtheorem{mydef}{Definition}[section]

\newtheorem{remark}[mydef]{Remark}



\newcommand{\R}{\mathbb{R}}

\newcommand{\Z}{\mathbb{Z}}

\newcommand{\eps}{\varepsilon}
\renewcommand{\epsilon}{\eps}

\newcommand{\wh}[1]{\widehat{#1}}

\newcommand{\poly}{\mathrm{poly}}
\newcommand{\polylog}{\mathrm{polylog}}

\newcommand{\E}{\operatorname{{\bf E}}}

\renewcommand{\hat}{\wh}

\newcommand{\T}{\mathrm{T}}

\newcommand{\calB}{\mathcal{B}}

\newcommand{\calT}{\mathcal{T}}

\newcommand{\bA}{\boldsymbol{A}}

\newcommand{\bC}{\boldsymbol{C}}

\newcommand{\bS}{\boldsymbol{S}}

\newcommand{\bX}{\boldsymbol{X}}
\newcommand{\bY}{\boldsymbol{Y}}

\newcommand{\ba}{\boldsymbol{a}}
\newcommand{\bb}{\boldsymbol{b}}

\newcommand{\ff}{\boldsymbol{f}}

\newcommand{\bx}{\boldsymbol{x}}
\newcommand{\by}{\boldsymbol{y}}

\newcommand{\ignore}[1]{}
\definecolor{cb-salmon-pink}{RGB}{255, 182, 119}
\definecolor{ref-color}{RGB}{200, 0, 200}

\DeclarePairedDelimiterX\diverg[2]{(}{)}{#1 \mathrel{}\mathclose{}\delimsize\|\mathopen{}\mathrel{} #2}

\crefname{enumi}{Step}{Steps}

\newcommand{\rank}{\mathrm{rank}}

\begin{document}

\title{On Sketching Trimmed Statistics}

\author{Honghao Lin}
\affiliation{%
  \institution{Carnegie Mellon University}
  \city{Pittsburgh}
  \country{USA}}
\email{honghaol@andrew.cmu.edu}

\author{Hoai-An Nguyen}
\affiliation{%
  \institution{Carnegie Mellon University}
  \city{Pittsburgh}
  \country{USA}}
\email{hnnguyen@andrew.cmu.edu}

\author{David P. Woodruff}
\affiliation{%
  \institution{Carnegie Mellon University}
  \city{Pittsburgh}
  \country{USA}}
\email{dwoodruf@cs.cmu.edu}


\begin{abstract}
We study sketching trimmed statistics of a frequency vector, including the $F_p$ moment of the top-$k$ coordinates and of the trimmed-$k$ vector. Despite their natural role in robust analytics, this is the first time these problems have been studied in any sublinear space setting. 
For $p \in [0,2]$, we obtain $\poly(\log n/\eps)$-space algorithms for both tasks when $k$ is moderately large, and for general $k$ we identify a sharp structural threshold that characterizes exactly when sublinear space is possible:
in particular, it is actually determined by the ratio between $a_k^2$ and $\|\mathbf{x}_{-k}\|_2^2/k$.
We extend these results to $p > 2$ and present several applications including algorithms for thresholded $F_p$ estimation and generalized impact indices. Notably, we improve the space bounds of Govindan, Monemizadeh, and Muthukrishnan (PODS ’17) for computing the $h$-index. 
\end{abstract}



\received{December 2025}
\received[accepted]{February 2025}

\maketitle

\section{Introduction}

The data-stream model has become a central abstraction for analyzing massive datasets that are too large to store explicitly. Examples include internet traffic logs, financial transactions, database logs, and scientific data streams (e.g., large-scale experiments in fields such as particle physics, genomics, and astronomy). Formally, we assume an underlying frequency vector $\bx \in \Z^n$
that receives a sequence of positive and negative updates to its coordinates.

Classical work beginning with Alon, Matias, and Szegedy~\cite{AMS1999space} introduced streaming algorithms for estimating global statistics such as the frequency moments $F_p = \sum_i |x_i|^p$, leading to an extensive line of research~\cite{W2004optimal, KNW2010optimal, I2006stable, KNW2010exact, NY2022optimal, BZ2024optimality, BJKS2004information, CKS2003near, IW2005optimal, LW2013tight, CCF2002finding}. Computing $F_p$ moments remains fundamental in database systems, powering tasks in query optimization, data mining, and network monitoring.

However, many analytics tasks in databases and data science rely not on global norms but on \emph{order-based} or \emph{trimmed} statistics---quantities that depend only on the largest, smallest, or central portion of the data. Such measures arise in robust statistics (where extremes are treated as outliers), in descriptive analytics, in regression and optimization (e.g., via Ky-Fan norms~\cite{CW2015sketching}), and in impact metrics such as the $h$-index. Despite their ubiquity, essentially nothing was known about whether such order-restricted statistics can be approximated in sublinear space.

We study two fundamental trimmed statistics.
The first is the \emph{top-$k$ $F_p$ moment}, defined as the $F_p$ moment of the $k$ largest coordinates of $\bx$. This statistic is also known as the Ky-Fan-Norm and has applications in generalized $\ell_p$ regression~\cite{CW2015sketching}, composite super-quantile optimization~\cite{RC2023algorithms}, frequency analysis, and descriptive statistics. Surprisingly, even for natural settings such as estimating the $F_p$ moment of the largest $n/2$ coordinates, no prior sublinear-space streaming algorithms were known.

The second statistic is the \emph{trimmed-$k$ $F_p$ moment}, obtained by discarding both the largest and smallest $k$ coordinates before computing the moment. Trimmed moments are classical tools for mitigating outliers and heavy-tailed noise, and they appear in robust estimation, exploratory analysis, and trimmed regression. Their practical importance is illustrated by the fact that statistical libraries (e.g., in \textsf{R}) include built-in functions for trimmed means and variances, which correspond to trimmed-$k$ $F_p$ norms for $p=1$ and $p=2$.

At first glance, these problems might appear solvable using standard tools such as Count-Sketch. However, such approaches require $\Theta(k)$ space to identify the top $k$ items, which is much larger than the $\mathrm{poly}(\log n, 1/\varepsilon)$-space bounds typical for $F_p$-moment estimation (for $p \in [0,2]$). Moreover, without structural assumptions on the frequency vector, heavy tails make it difficult to isolate top entries with sufficient accuracy. These limitations motivate the central question of this work:
\begin{quote}
\emph{Can we design sublinear-space, one-pass streaming algorithms that approximate trimmed $F_p$ statistics---specifically, the $F_p$ moments of the top-$k$ and trimmed-$k$ versions of the frequency vector?}
\end{quote}

We answer this question affirmatively. We design the first sublinear-space \emph{linear sketches} for approximating top-$k$ and trimmed-$k$ $F_p$ statistics, which can be applied directly to a stream of updates. Moreover, linearity ensures that our algorithms extend naturally to distributed and mergeable settings, while also supporting fast update times.

\subsection{Our Contributions}
We note that the update time of all of our algorithms, or the time that our algorithm takes to process an update, is $\poly(\log n, 1/\eps)$. 
\paragraph{First sublinear-space algorithms for trimmed statistics.}
We design the first linear sketches for approximating both the top-$k$ and trimmed-$k$ frequency moments for $p\in[0,2]$ up to a $(1+\eps)$ multiplicative approximation in space $\poly(\log n/\eps)$
whenever $k$ is moderately large (e.g.,
$k \ge n / \polylog n$).

\paragraph{A sharp structural characterization.}
For general $k$, we identify an explicit condition involving the tail of $\bx$ under which sublinear space
is possible. Where $\ba$ is $\bx$ sorted in decreasing order by absolute value and $\bx_{-k}$ is $\bx$ without the top $k$ coordinates, we show that approximating the top-$k$ $F_p$ moment using
$\poly(\log n/\varepsilon)$ space is possible \emph{if and only if}
\[
  a_k^2 \;\gtrsim\; \poly\left(\frac{\varepsilon}{\log n}\right) \cdot \frac{\|\bx_{-k}\|_2^2}{k}.
\]
Above this threshold, our sketches give $(1\pm\varepsilon)$-approximations;
on the other hand, when $a_k^2 \le \frac{k}{n^{c}} \cdot \frac{\|\bx_{-k}\|_2^2}{k}$, we prove $n^{\Omega(1)}$ space is necessary.
This yields a sharp complexity characterization whose threshold is determined by the ratio between $a_k^2$ and $\frac{\|\bx_{-k}\|_2^2}{k}$.
An analogous characterization holds for the trimmed-$k$ statistic, though we show that an additive
error term $\varepsilon  k\cdot |a_{k - \varepsilon k}|^p$ is unavoidable. 

To the best of our knowledge, no prior work has established a connection between trimmed statistics and conditions on the $k$-residual error. 

We also regard our analysis as a concrete contribution. Our algorithms employ a multi-level sub-sampling scheme, where we identify heavy hitters at each level and leverage this information to estimate the value of the desired frequencies. Unlike prior work using similar frameworks, we require a refined analysis of the level set structure, specifically by comparing the residual norms at different levels. We give a more comprehensive view in our technical overview.

\paragraph{Extending to $p > 2$ and applications.}
We extend our results to $p>2$, incurring a $\Theta(n^{1-2/p})$ factor that is provably necessary when estimating $F_p$ globally~\cite{LW2013tight}.

We also remark on two direct applications of our algorithms. 
\begin{itemize}
    \item Li, Lin, and Woodruff~\cite{LLW2024optimal} give an $O(k^{2/p} n^{1-2/p}\,\mathrm{poly}(\log n/\eps))$-space linear sketch for estimating the $k$-residual error $\|\bx - \bx_k\|_p^p$ for $p>2$, where $\bx_k$ is the best $k$-sparse approximation of $\bx$. This residual quantifies the benefit of using a more computationally expensive sparse approximation (i.e., increasing $k$). When $\|\bx_k\|_p^p = \Theta(1)\|\bx_{-k}\|_p^p$, a $(1\pm\eps)$ estimate of the top-$k$ moment yields a $(1\pm\eps)$ estimate of the $k$-residual error, and in this regime our algorithm achieves a $k^{2/p}$ improvement over~\cite{LLW2024optimal}, which is substantial for large $k$.
    \item  Estimating the $F_p$ moment of a trimmed-$k$ vector extends naturally to trimmed regression. Given a matrix $\mathbf{A}$ and vector $\mathbf{b}$, trimmed regression minimizes $\|\mathbf{A}x - \mathbf{b}\|_p^p$ after removing the top and bottom $k$ residuals. A linear sketch for trimmed norms allows efficient estimation of $\|\mathbf{A}x - \mathbf{b}\|_{p,\mathrm{trim}}$ for any $x$. This enables the computation of the trimmed norm across all possible $\bx$ by leveraging a well-constructed net argument. 
\end{itemize}

We also obtain more involved applications that require not only estimating the contribution of the top-$k$ coordinates but also identifying the appropriate value of $k$ itself.
\begin{itemize}
  \item Thresholded $F_p$ estimation: approximating $\sum_{i \in \calB_{\calT}} |x_i|^p$ where $\calB_{\calT}$ is the set of coordinates larger than input $\calT$. This is closely related to heavy hitters which is a well-studied problem in databases and streaming \cite{BCIW2016beating, BDW2016optimal, BCINWW2017bptree, WMZ2016new}. We note that we can take the threshold to be $\eps \norm{\bx}_p$ and in parallel calculate $\norm{\bx}_p$ to get the $F_p$ moment of the $\ell_p$ heavy hitters. 
  \item Moment based thresholding: approximating $k$, where $k$ is the largest integer such that $\sum_{i=0}^k |a_i|^p \ge k^{p+1}$, and $\ba$ denotes $\bx$ sorted in decreasing order. For $p = 1$ this is equivalent to computing the $g$-index \cite{E2006theory}.
  \item Value based thresholding: approximate the top-$k$ $F_p$ moment where $k$ is the largest integer such that each of the top $k$ frequencies is at least $k$. For $p = 0$, this corresponds to the popular $h$-index~\cite{H2005index} and for $p = 1$ dividing this by $k$ corresponds to the $a$-index \cite{ACHH2009hindex}. 
\end{itemize}

The $h$-index, $g$-index, and $a$-index can all be viewed as impact indices. Beyond measuring academic publishing impact, the $h$-index has been used to identify influential users in social networks~\cite{R2015measuring} and to reveal structural properties of large networks, such as approximating the network’s degree distribution, computing node coreness, and supporting truss and nucleus decomposition for identifying dense subgraphs~\cite{EJPRS2018provable, LZZS2016hindex, SSP2018local}.

Our algorithm improves upon the $h$-index algorithm of Govindan, Monemizadeh, and Muthukrishnan~\cite{GMM2017streaming}, which requires a parameter $\beta$ that lower bounds the $h$-index and uses $O(\poly \log n \cdot \frac{n}{\beta \eps^2})$ bits of space. Setting $\beta = 1$ yields linear space, and even with a strong lower bound the savings remain limited. For instance, when $h = n^{1/3}$ their algorithm still requires $O(n^{2/3})$ bits, whereas our algorithm achieves a $(1\pm\eps)$ approximation using only $O(\poly(\log n/\eps))$ bits.

\paragraph{Experimental results.}
We illustrate the practicality of our algorithms by running experiments on two real world datasets and one synthetic dataset. Specifically, we compare our algorithm for computing the $F_p$ moment for the top-$k$ frequencies against the classical Count-Sketch and show that our algorithm achieves better accuracy when given the same space allotment.

\subsection{Extended Related Work}
There has been a long line of work on computing the frequency moment, or $F_p$, of a dataset. See \Cref{table:fp} for a picture of previous work. 
\begin{table}[h] 
\centering
\begin{tabularx}{\textwidth}{|X|X|c|c|}
\hline
\textbf{$F_p$} & \textbf{Stream} & \textbf{Upper bound} & \textbf{Lower bound} \\ \hline

$p = 0$ & insertion-only & $O(\eps^{-2} + \log n)$ \cite{KNW2010optimal} &  $\Omega(\eps^{-2} + \log n)$ \cite{W2004optimal}\\ \hline

$p = 0$ & turnstile & $O(\eps^{-2} \log n \log\log n)$ \cite{KNW2010optimal} & $\Omega(\eps^{-2}\log n \log\log n)$ \cite{DMWY2023separating}
\\ \hline

$p = 1$ & insertion-only & $O(\log \log n + \log \eps^{-1}) \cite{NY2022optimal}$ & $\Omega(\log \log n + \log \eps^{-1})$ \cite{NY2022optimal}
\\ \hline

$0 < p < 2$ & insertion-only, turnstile & $O(\eps^{-2} \log n)$ \cite{KNW2010exact} & $\Omega(\eps^{-2} \log n)$ \cite{KNW2010exact}\\ \hline

$p = 2$ & insertion-only, turnstile & $O(\eps^{-2} \log n )$ \cite{AMS1999space} & $\Omega(\eps^{-2} \log n)$ \cite{BZ2024optimality} \\ \hline
$p > 2$ & turnstile & $O(n^{1-2/p}(\eps^{-2} + \eps^{-4/p} \log n))$ \cite{G2011polynomial}& $\Omega(\eps^{-2} n^{1-2/p} \log n)$ \cite{LW2013tight}\\
\hline
\end{tabularx}
\caption{Upper and lower bounds for $F_p$ estimation. All space bounds are in bits.}
\label{table:fp}
\end{table}
There is also a long line of work on $\ell_p$ sampling, which is closely related. In this problem, we are given frequency vector $\bx \in \Z^n$ and the goal is to return an index $i \in \{1,2,\ldots, n\}$ with probability $|x_i|^p / \| \bx \|_p^p$. There are generally two samplers that are studied: approximate \cite{MW2010pass, AKO2010streaming, JST2011tight} and perfect \cite{FIS2008sampling, JW2021perfect, WXZ2025perfect}. An approximate sampler is one that returns an index $i$ with probability $(|x_i^p| / \| \bx \|_p^p)  \cdot (1 \pm \eps) + \poly(1/n)$ for $\eps \in (0,1)$. When we have an approximate sampler where $\eps = 0$, we call these ``perfect" samplers. 
The best known lower bound is $O(\log^3 n)$ (including a $\log n$ factor from the failure probability) for turnstile streams by Kapralov, Nelson, Pachocki, Wang, Woodruff, and Yahyazadeh~\cite{KNPWWY2017optimal}.
See \cite{JW2021perfect} for a more thorough discussion and table on previous work. 

We also discuss prior work on techniques and statistics related to the ones we study. We again emphasize that we are not aware of any work on the top-$k$ and trimmed-$k$ statistics that we study in this paper. 
Daliri, Freire, Musco, Santos, and Zhang~\cite{DFMSZ2024sampling} study inner product estimation via linear sketches, using priority sampling to preserve large coordinates. Although their focus is different, the underlying sampling idea is related to techniques for trimmed statistics, where one must isolate the contribution of heavy coordinates. Cohen and Kaplan~\cite{CK2007bottom} studied bottom-$k$ sketches. Here, there is a ground set where each item in the set is given an independent random rank dependent on the weight of the item. Then, for some input subset, a bottom-$k$ sketch contains the $k$ items with the lowest ranks. This relates to our setting, where we also study approximations of statistics defined by restricting our attention to subsets of the data. 

\subsection{Road Map}
In \Cref{sec:prelim} we have our preliminaries. In \Cref{sec:TO} we give a technical overview. In \Cref{sec:topk} we present our algorithm for computing the $F_p$ for $p \in [0,2]$ for the top-$k$ vector. In \Cref{sec:trimkbody} we present our algorithm for computing the $F_p$ for $p \in [0,2]$ of the trimmed-$k$ vector. In \Cref{sec:p2} we present the extension of our trimmed statistic algorithms to $F_p$ for $p>2$. In \Cref{sec:apps} we present a few applications of our algorithms (which include algorithms for impact indices $h$, $g$, and $a$). 
In \Cref{sec:LBs} we give our hardness results. In \Cref{sec:exper} we give our experiments. 

\section{Preliminaries} \label{sec:prelim}
\paragraph{Notation.}
We use $x_i$ to denote the $i^{\text{th}}$ entry of input vector $\bx$ or equivalently the frequency of element $i$. $\rank(v)$ denotes the rank of item $v$ in vector $\bx$, or the number of entries in $\bx$ with value \emph{greater} than $v$. $\bx_{-k}$ denotes the vector $\bx$ excluding the top $k$ frequencies by absolute value. $|\bx|$ denotes the length/dimension of vector $\bx$. In general, we boldface vectors and matrices. 

\paragraph{Count-Sketch and heavy hitters.}
\label{sec:count-sketch}
We review the Count-Sketch algorithm~\cite{CCF2002finding}. We have $q$ distinct hash functions $h_i: [n] \to [B]$ and an array $\bC$ of size $q \times B$. Additionally, we have $q$ sign functions $g_i: [n] \to \{-1, 1\}$. The algorithm maintains $\bC$ (a Count-Sketch structure) such that $C[\ell, b] = \sum_{j: h_\ell(j) = b} g_{\ell} (j) \cdot x_j$. The frequency estimation $\hat{x}_i$ of $x_i$ is defined to be the median of $\{g_\ell(i) \cdot C[\ell, h_\ell(i)]\}_{\ell \le q}$. Here, the parameter $B$ is the number of the buckets we use in this data structure. Formally, when $q = O(\log n)$, we have with probability at least $1 - 1/\poly(n)$, $|\hat{x}_i - x_i| \le O\left(\frac{\norm{\bx_{-B}}_2}{\sqrt{B}}\right)$. Based on this, we obtain the following Heavy Hitter data structure.

\begin{definition}
    Given parameters $\theta, k$, a \textsf{HeavyHitter} data structure $\mathcal{D}$ that receives a stream of updates to the frequency vector $\ff$ and provides a set $T \in [n]$ of heavy hitters of $\ff$, where
    \begin{enumerate}
        \item $i \in T$ if $f_i \geq \theta ||\ff_{-k}||_2,$
        \item For every $i \in T$, we have $f_i \ge \frac{9}{10}\cdot \theta ||\ff_{-k}||_2$
    \end{enumerate}
     where $\ff_{-k}$ is the frequency vector excluding the top $k$ frequencies. Moreover, for every $i \in T$, $\mathcal{D}$ can estimate $f_i$ up to $\frac{1}{k} \norm{\ff_{-k}}_2$ additive error. 
\end{definition}

\paragraph{Turnstile streaming model.}
Our input is an $n$-dimensional frequency vector $\bx$, initially all zeros. The algorithm processes a stream of updates of the form $(i, \pm 1)$, which increment or decrement $\bx_i$. Updates arrive in arbitrary order, and we assume the stream has length at most $\poly(n)$. The goal is to process the stream using sublinear space and only one pass. As is standard, we assume each coordinate is bounded above by an integer $m$ throughout the stream.

\paragraph{Linear sketches.} \label{prelim:LS}
We can compress an $n$-dimensional vector $\bx$ by multiplying it with a matrix $\bS \in \R^{r \times n}$. A matrix $\bS$ is a linear sketch if it is drawn independently of $\bx$ and if $\bS(x + c_i) = \bS x + \bS c_i$ for any update $c_i$ that changes a single entry of $\bx$ by $\pm 1$. Independence allows $\bS$ to be generated without knowledge of $\bx$, and linearity ensures that updates can be incorporated without storing $\bx$ explicitly. In practice, $\bS$ is stored implicitly using pseudorandom hash functions, enabling efficient updates. The goal is to minimize the required space, ideally making $r << n$. 

\paragraph{Subsampling scheme $P$.}
In our algorithms we require a subsampling scheme $P$ such that the $i$-th level has subsampling probability $r_i = 2^{-i}$ to each coordinate of the underlying vector. A hash function is used to remember which coordinates are subsampled in each level. If coordinates $j_1, j_2, \cdots, j_w$ are subsampled in the $i^{\text{th}}$ subsampling level, we are not storing these coordinates explicitly. Rather, this subsampling level only looks at updates which involve $j_1, j_2, \cdots, j_2$ to update its structures (in our case, its heavy hitters structure).

\paragraph{Derandomization.}

Throughout our paper, we assume that our hash functions exploit full randomness. Such an assumption can be removed with an additional $\poly(\log n)$ factor, e.g., the use of Nisan’s pseudorandom generator~\cite{N1992pseudorandom} or its variants. We refer the readers to a more detailed discussion in~\cite{JW2021perfect} (Theorem 5).

\section{Technical Overview}
\label{sec:TO}
We begin with the problem of estimating the $F_p$ norm of the top $k$ frequencies, starting with the simple case where each of the top $k$ frequencies has value $a_k$. 
Suppose that the condition
\begin{equation} \label{eq:condition}
a_k^2 \ge \poly(\eps/\log n) \cdot \frac{\| \bx_{-k}\|_2^2}{k}
\end{equation}
holds. A naive approach is to run Count-Sketch (\Cref{sec:count-sketch}) with $O(k)$ buckets, whose tail error is 
$O\!\left(\frac{\|\bx_{-k}\|_2}{\sqrt{k}}\right) = O(a_k)$. This captures each of the top-$k$ items but uses $\Theta(k)$ buckets, which is suboptimal when $k$ is large. Our goal is to use only $\poly(1/\eps, \log n)$ bits of space.

To reduce the space, we first sub-sample the stream. Let $\hat{\bx}$ be the sub-sampled stream obtained by sampling each coordinate independently with probability
$p = \min(1, c/(\eps^2 k))$,
for a constant $c$. Then $pk = \Theta(1/\eps^2)$, so in expectation the sub-stream contains $\Theta(1/\eps^2)$ of the top-$k$ items. A Chernoff bound shows that after rescaling by $1/p$, the number of surviving coordinates in $\hat{\bx}$ that were in the top $k$ is $(1 \pm \eps) \cdot k$ with high probability.

The key point here is that sub-sampling greatly reduces the tail mass while preserving enough of the coordinates in the top $k$. This allows us to use a much smaller Count-Sketch. 
Specifically, let $k' = \poly(1/\eps, \log n)$ and apply Count-Sketch with $k'$ buckets to the sub-stream $\hat{\bx}$. The tail error in this sub-stream is
$O\left( \frac{\|\hat{\bx}_{-k'}\|_2}{\sqrt{k'}} \right).$
Since only a $p$-fraction of the tail survives, $\|\hat{\bx}_{-k'}\|_2$ shrinks by approximately $\sqrt{p}$, and thus the tail error is
\[
O\left( \frac{\|\hat{\bx}_{-k'}\|_2}{\sqrt{k'}} \right) = \poly(\eps/\log n) \cdot \frac{\|\bx_{-k}\|_2}{\sqrt{k}}.
\]
Combined with condition~\eqref{eq:condition}, this error is small enough to correctly isolate every surviving coordinate of the top $k$ in the sub-stream. After rescaling by $1/p$, we obtain a $(1\pm\eps)$ approximation to the $F_p$ norm of the top-$k$. On the other hand, when condition~\eqref{eq:condition} fails, a reduction to a variant of the 2-SUM problem of~\cite{CLLTWZ2024tight} shows that any such approximation requires $n^{\Omega(1)}$ bits of space.

We now turn to the general case. Here we use the level-set framework introduced by~\cite{IW2005optimal}. We divide the coordinates of $\bx$ into exponentially growing intervals 
\[
[0,1+\eps),\ [1+\eps,(1+\eps)^2),\ \dots,
\]
up to the maximum coordinate called level sets. 
At $O(\log n)$ decreasing subsampling rates, we maintain $\ell_2$ heavy hitters for each subsampled stream. For every level set that contributes significantly to the top-$k$ portion of the $F_p$ norm (i.e., whose size multiplied by its level value contributes non-negligibly to the $F_p$ mass), we locate a subsampling rate at which there are $\Theta(1/\varepsilon^2)$ expected survivors from that level. Using the stored heavy hitters for this sub-stream, we estimate the number of coordinates belonging to that level set. Summing these estimated contributions over all relevant levels yields a good approximation to the top-$k$ $F_p$ moment.

The algorithmic framework is standard, but our analysis is not. Prior work such as~\cite{IW2005optimal} analyzes the tail norm at a single truncation level and does not need to compare multiple truncation levels simultaneously. In contrast, our setting requires understanding how the residual norms $\|x_{-t}\|_2$ behave across many adjacent values of $t$. A naive approach fails because the sketching noise can mask the small differences between $\|x_{-t}\|_2$ and $\|x_{-(t+1)}\|_2$. Our main technical contribution is a refined multi-level analysis that ensures these comparisons are stable under sketching noise.

To explain this structure, consider a contributing level set $[v = (1+\eps)^i,\ (1+\eps)^{i+1})$, 
and let $s_i$ be the number of coordinates in this interval. Let $\mathrm{rank}(v)$ be the number of coordinates greater than $v$. We show that condition~\eqref{eq:condition} implies the multi-level analogues
\[
v^2 s_i \ge \poly(\eps/\log n)\cdot \|x_{-\mathrm{rank}(v)}\|_2^2,
\qquad
s_i \ge \poly(\eps/\log n)\cdot \mathrm{rank}(v).
\]
These inequalities state that the coordinates in this level set carry enough $\ell_2$ mass to dominate the tail after $\mathrm{rank}(v)$. This allows us to run a sub-sampled Count-Sketch with sampling probability
$p = \min(1, c/(\eps^2 s_i)),$
so that the expected number of survivors from this level is $\Theta(1/\eps^2)$. The Count-Sketch tail error under this sampling is
\[
O\left( \frac{ \|x_{-\mathrm{rank}(v)}\|_2 }{\sqrt{s_i}} \right) \cdot \poly(\eps/\log n),
\]
which is small compared to $v$, allowing us to identify all important coordinates in this level. After rescaling, this gives a $(1\pm\eps)$ approximation of $s_i$. To our knowledge, this type of multi-level residual comparison has not appeared before and is the core analytical novelty of our work.

We next consider estimating the $F_p$ moment of the trimmed-$k$ vector, obtained by removing the largest and smallest $k$ coordinates. Naively, we could estimate this by taking the difference between the $F_p$ moment of the top $(n-k)$ frequencies and the $F_p$ moment of the top $k$ frequencies which has error $O(\eps) \cdot (\sum_{i = 1}^{n - k} |a_i|^p)$. 
To achieve better error, we rely on a key observation: we can compute both the top-$(n-k)$ and top-$k$ moments using the \emph{same} level-set decomposition and the same estimates $\tilde{s}_j$ for the sizes of all level sets. This means that when the estimator identifies the position where the first $k$ items end, both computations use the \emph{same approximate cutoff}. Because of this,  the large contribution of the top $k$ coordinates to the error cancels when we subtract the two estimates. In particular, what we are actually estimating is 
$\sum_{i=u}^{u+n-2k} |a_i|^p$
for some $u = k \pm O(\eps k)$, corresponding to the middle portion of the vector after discarding the top $k$ items up to the available accuracy.
Since each contributing level set is estimated to within a $(1\pm\eps)$ factor, the resulting estimate of this block incurs a total error of
$\eps\left( \sum_{i=k}^{n-k} |a_i|^p \;+\; k|a_{k-\eps k}|^p \right).$
We also give a hard instance to show that this extra additive error term $\eps \cdot k|a_{k - \eps k}|^p$ is unavoidable. 

\section{$F_p$ Moment for $p \in [0,2]$ of Top-$k$ Frequencies} \label{sec:topk}

Here, we prove the following theorem with an algorithm that estimates the $F_p$ moment for $p\in[0,2]$ of the top $k$ frequencies. 

\begin{theorem}[$F_p$ of top-$k$ vector for $0 \leq p \leq 2$] \label{thm:top_k}
Given $\bx \in \Z^n$,  $ p \in [0,2]$, $k \geq 0$, and $\eps \in (0,1)$, suppose that we have 
    \begin{equation}
    \label{eq:1bb}
    a_k^2 \ge (\eps/\log n)^c \cdot \frac{\| \bx_{-k}\|_2^2}{k}.
    \end{equation}
    for some constant $c$. Then there exists a linear sketch that uses $\poly(\log n / \eps)$ bits of space (where the exponent depends on the value of $c$) and estimates $\sum_{i = 1}^k |a_i|^p$ up to a $(1 \pm \eps)$ multiplicative factor with high constant probability.
\end{theorem}

\begin{remark}
    When $k \geq \polylog n$, condition~(\ref{eq:1bb}) is always met. Therefore, this directly gives us the result for moderately large $k$. 
\end{remark}

Before presenting the full algorithm, we first give our \Cref{alg:levelSetEst} which estimates the size of each ``contributing" level set that contains at least one top-$k$ item up to a $(1 \pm \eps)$-factor. We then give the formal definition of what it means for a level set to ``contribute" and show that accurately estimating the sizes of contributing level sets gives a good final approximation. 

\begin{algorithm}[t]
\caption{Level-Set-Estimator ($\eps \in (0, 1]$)}
\label{alg:levelSetEst}
\begin{algorithmic}[1] 
\REQUIRE $(i)$ A subsampling scheme $P$ such that the $i$-th level has subsampling probability $r_i = 2^{-i}$ to each coordinate of the underlying vector; $(ii)$ an upper bound on each coordinate $m$; $(iii)$ $L + 1$ HeavyHitter structures $\mathcal{D}_0, \ldots, \mathcal{D}_L$ with parameter $\theta= (\eps / \log n)^{c + 9}$ where $L = \log n$ and  $\mathcal{D}_i$ corresponds to the $i$-th substream. 
\STATE $t_0 \leftarrow K \log(\eps^{-1}\log m) $ where $K$ is a large constant. 
\STATE $\zeta \leftarrow$ uniform random variable between $[1/2, 1]$. 
\STATE $t \leftarrow \log_{1+\eps}(m) + 1$. 
\FOR{$j = 0, \ldots, L$}
    \STATE $\Lambda_j \leftarrow \poly(\eps/L)$ heavy hitters from $\mathcal{D}_j$. 
\ENDFOR
\FOR{$j = 0, \ldots, t_0$}
    \STATE Let $\tilde{s_j}$ be the number of elements contained in $\Lambda_0$ in $ [\zeta (1+\eps)^{t-j-1}, \zeta (1+\eps)^{t-j}]$. \label{line:8}
\ENDFOR 
\FOR{$j = t_0 + 1,..., t-1$}
    \STATE Find the largest $\ell$ s.t. $\Lambda_\ell$ contains $z = \Theta(\log n/\eps^2)$ elements 
    in $[\zeta (1+\eps)^{t-j-1}, \zeta (1+\eps)^{t-j}]$. 
    \IF{such $\ell$ exists} 
        \STATE $\tilde{s_j} \leftarrow z \cdot 2^\ell.$
    \ELSE 
        \STATE $\tilde{s_j} \leftarrow 0$. 
    \ENDIF
\ENDFOR
\STATE Output each $\tilde{s_j}$ for all $j \in [t]$. 
\end{algorithmic}
\end{algorithm}

\begin{algorithm}[t]
\caption{Ky-Fan-$k$-Norm-Estimator ($\eps \in (0, 1]$, $p \in (0, 2]$, $k \geq 0$)}
\label{alg:topk}
\begin{algorithmic}[1] 
\STATE Run \hyperref[alg:levelSetEst]{Level-Set-Estimator ($\eps$)}.
\STATE Find the $i$ such that $\sum_{j=0}^{i - 1} \tilde{s}_j < k$ and $\sum_{j=0}^{i} \tilde{s}_j \ge k$. 
\STATE \textbf{Return} $\left(\sum_{j=0}^{i - 1} \tilde{s_j} \cdot \zeta^p (1+\eps)^{p(t-j)}\right) + \left(k - \sum_{j = 0}^{i - 1} \tilde{s}_j\right) \zeta^p(1+\eps)^{p(t-i)}$. 
\end{algorithmic}
\end{algorithm}

\begin{definition} \label{def:contribute}
    Suppose that $m$ is an upper bound on $\|\bx\|_{\infty}$ and $t = 1 + \log_{1 + \eps} m$. Let $\zeta$ be a uniform random variable in $[1/2, 1]$. Define the level set $S_j$ for $j \in [1, \log_{1+\eps}m]$ to be 
    \[
    S_j = \{i \in [n]: |x_i| \in [\zeta (1+\eps)^{t-j-1}, \zeta (1+\eps)^{t-j})\} 
    \]
    and $s_j = |S_j|$. We say the level set $S_j$ ``contributes" if
    \[
    \sum_{i \in S_j} |x_i|^p \ge \frac{\eps^2}{\log^2 n} \left(\sum_{i = 1}^k |a_i|^p\right).
    \]
\end{definition}

\begin{lemma} \label{lem:condition}
     Suppose that $a_k^2 \ge (\eps/\log n)^c \cdot \frac{\norm{\bx_{-k}}_2^2}{k}$. For any $j$ such that $\zeta (1+\eps)^{t-j} \ge |a_k|$ and $S_j$``contributes", 
    taking $v = \zeta(1+\eps)^{t - j - 1}$ we have $v^2 \cdot s_j \ge (\eps / \log n)^{c + 5} \cdot \norm{\bx_{-\mathrm{rank}(v)}}_2^2$
    where $\mathrm{rank}(v)$ is the rank of $v$ in array $\bx$ (i.e., the number of entries in $\bx$ with value greater than $v$).
\end{lemma}
\begin{proof}
    Since $S_j$ contributes, we have 
    \[
    s_j v^p \geq (\eps / \log n)^2 k a_k^p. 
    \]
    Multiplying both sides by $v^{2-p}$ gives 
    \begin{equation} \label{eq:assump1}
    s_j v^2 \geq (\eps / \log n)^2 k a_k^p v^{2-p} \geq (\eps / \log n)^2 k a_k^2 
    \end{equation}
    where the last inequality is because we have $v^{2-p} \geq a_k^{2-p}$ since $v \geq a_k$ and $p \leq 2$. Now using the lemma assumption $a_k^2 \geq (\eps / \log n)^c \cdot \| \bx_{-k}\|_2^2 / k$ and multiplying by $k$ gives 
    \[
    k a_k^2 \geq (\eps / \log n)^c \| \bx_{-k}\|_2^2. 
    \]
    Substituting this into (\ref{eq:assump1}) gives us 
    \begin{equation} \label{eq:assump2}
        s_j v^2 \geq (\eps / \log n)^{c+2} \|\bx_{-k}\|_2^2. 
    \end{equation}
    We next measure the difference between $\norm{\bx_{-k}}_2^2$ and $\norm{\bx_{-\mathrm{rank}(v)}}_2^2$. Define the level set $T_\ell$ for $\ell \in [0,q-1]$ where
\[
T_\ell = \{i \in [n]: |x_i| \in [v/2^{\ell + 1}, v/2^\ell)\} , 
\]
for $v/2^{\ell + 1} \ge |a_k|$ and 
\[
T_q = \{i\in [n]: |x_i| \in [|a_k|, v/2^q)\} \;. 
\]
Let $t_\ell = |T_\ell|$, $t_q = |T_q|$. 
Now, we claim that  we have for $w \in [0, q]$ that $t_w \le O(\log^2 n / \eps^2) \cdot s_j \cdot 2^{wp}$. Otherwise we would have 
\[
\left(\sum_{i = 1}^k |a_i|^p\right) 
\ge \left(\frac{v}{2^{w + 1}}\right)^{p} \cdot t_w
\ge \left(\frac{v}{2^{w + 1}}\right)^{p} \cdot O\left(\frac{\log^2 n}{\eps^2}\right) s_j 2^{wp} \ge O\left(\frac{\log^2 n}{\eps^2}\right) \left(\sum_{i \in S_j} |x_i|^p\right) \;
\]
which contradicts the fact that $S_j$ contributes. 
So now, taking a sum we get 
\[
\sum_{\mathrm{rank}(v) \le i \le k} a_i^2 \le \sum_{i \in [q]} t_i  \left(\frac{v}{2^{i}}\right)^2 \leq \sum_{i \in [q]} \left( \frac{v}{2^i} \right)^2 O\left(\frac{\log^2 n}{\eps^2}\right) s_j 2^{ip} \le v^2 s_j \cdot O\left(\frac{\log^3 n}{\eps^2}\right) \;.
\]
This implies that 
\begin{equation} \label{eq:rank}
\norm{\bx_{-k}}_2^2 \ge \norm{\bx_{-\mathrm{rank}(v)}}_2^2 - v^2 s_j \cdot O\left(\frac{\log^3 n}{\eps^2}\right) \;.
\end{equation}
Combining~\eqref{eq:assump2} and~\eqref{eq:rank} we immediately get that
\[
v^2 \cdot s_j \ge (\eps / \log n)^{c + 5} \cdot \norm{\bx_{-\mathrm{rank}(v)}}_2^2 \;. \qedhere
\]
\end{proof}

\begin{lemma} \label{lem:sjapprox}
    Consider some level set $S_j$ with $s_j = |S_j|$. Consider the subsampling stream $\mathcal{P}$ where each coordinate is sampled with probability $r = \min\left(1, \frac{C \log n}{s_j \eps^2} \right)$ for some constant $C$ and suppose that $S_j$ has $z$ survivors in this stream. Let $\tilde{s}_j = z / r$, then with probability $1 - 1/\poly(n)$ we have $(1 - \eps)s_j \le \tilde{s}_j \le (1 + \eps)s_j$. 
\end{lemma}
\begin{proof}
Denote the coordinates in $S_j$ as $u_i$ for $i \in [s_j]$. Let $X_i$ for $i \in [s_j]$ denote an indicator random variable which is $1$ if $u_i$ was sampled in $\mathcal{P}$ and $0$ otherwise. We have that $\E[X_i] = r$ and $\E\left[\sum_i X_i\right] = s_j r$. Then, from Chernoff's bound we have that 
\[
\mathbf{Pr} \left[\left|\sum_i X_i - s_j r\right| \ge \eps s_j r \right] \le 2\exp(-\eps^2 s_j r / 3) \leq 2\exp(-C \log n/3) \le 1 / \poly(n). \qedhere
\]

\end{proof}

\begin{lemma} \label{lem:lowSL}
    Consider some level set $S_j$ with $s_j = |S_j|$. Consider the subsampling stream $\mathcal{P}/2$ where each coordinate is sampled with probability $r = \min\left(1, \frac{C \log n}{s_j \eps^2} \cdot \frac{1}{2} \right)$ for some constant $C$. With probability at least $1-1/\poly(n)$ there are less than $c\log n/\eps^2$ survivors from $S_j$. 
\end{lemma}
\begin{proof}
    Denote the coordinates in $S_j$ as $u_i$ for $i \in [s_j]$. Let $X_i$ for $i \in [s_j]$ denote an indicator random variable which is $1$ if $u_i$ was sampled in $\mathcal{P}/2$ and $0$ otherwise. We have that $\E[X_i] = r$ and $\E[X] = s_j r$ for $X = \sum_i X_i$. From Chernoff's bound we have that 
    \[
    \mathbf{Pr}\left[X \geq \frac{c \log n}{\eps^2}\right ] = \mathbf{Pr}[X > 2 \cdot \E[X]] \leq 2 \exp (-C \log n / 6) \leq 1/\poly(n).  \qedhere
    \]
\end{proof}

\begin{lemma}
    \label{lem:sampling_rate}
    Suppose that $a_k^2 \ge (\eps/\log n)^c \cdot \frac{\norm{\bx_{-k}}_2^2}{k}$ and $S_j$ contributes where $\zeta (1+\eps)^{t-j} \ge |a_k|$. Then, with probability at least $1 - \poly(\eps / \log n)$, we have $\tilde{s}_j \in [(1 - \eps)s_j, (1 + \eps)s_j]$.
\end{lemma}

\begin{proof}
    Consider a level set $S_j$ with a value range in $[v, (1 + \eps) v]$ where $\zeta (1+\eps)^{t-j} \ge |a_k|$. Consider the sub-stream $\mathcal{P}$ with corresponding sampling rate $r  = \min\left(1, \frac{C \log n}{s_j \eps^2}\right)$ for a sufficiently large constant $C$ and assume there are $z$ survivors of $S_j$ in $\mathcal{P}$.
    Since we have a random threshold $\zeta$ in the boundary of the level set, we have with probability at least $1 - \poly(\eps /\log n)$, a $(1 - \eps)$ fraction of them is within $[v\left(1 + \poly(\eps / \log n)\right), 
(1 + \eps) v \left(1 - \poly(\eps / \log n)\right)]$. 

    We next analyze the tail error of the heavy hitter data structure. First we have that 
    \begin{equation} \label{eq:sjRank}
    s_j \ge (\eps^2 / \log^2 n) \cdot \rank(v).
    \end{equation}
    Suppose that we had $s_j < (\eps^2 / \log^2 n) \cdot \rank(v)$. This would mean that $\rank(v) \cdot v^p \geq s_j (\log^2 n / \eps^2) \cdot v^p$, which contradicts the fact that $S_j$ 
    contributes. Therefore, combining 
    (\ref{eq:sjRank}) with the fact that  
    $\mathcal{P}$ 
    has sampling rate 
    $r = \min\left(1, \frac{C \log n}{s_j \eps^2}\right)$ gives us that with probability $1- \poly(\eps/\log n)$ using Chernoff's bound that the number of survivors of the top $\rank(v)$ coordinates of $x$ surviving in $\mathcal{P}$ is at most $O((\log n / \eps)^4)$. 
    
    Recall that we use $(\log n/\eps)^{c + 9}$ buckets in our heavy hitter data structure (\Cref{sec:count-sketch}). Hence, with probability at least $1 - \poly(\eps /\log n)$, the tail error of the heavy hitter data structure will be at most $\frac{1}{\sqrt{s_j}} (\eps / \log n)^{c/2 + 9/2} \cdot \norm{\bx_{-\rank(v)}}_2$. 
    
    Recall that we have condition $a_k^2 \ge (\eps/\log n)^{c} \cdot \frac{\norm{x_{-k}}_2^2}{k}$ and therefore from Lemma~\ref{lem:condition} have $v^2 s_j \ge (\eps / \log n)^{c + 5} \cdot \norm{\bx_{-\mathrm{rank}(v)}}_2^2$. 
     Combining this with the above tail error we immediately have with high probability the heavy hitter data structure can identify every survivor in $[v\left(1 + (\eps / \log n)^2\right), 
(1 + \eps) v \left(1 - (\eps / \log n)^{2}\right)]$. 

Combining this with Lemma~\ref{lem:sjapprox}, the remaining thing is to show that with high probability that a level with a smaller sampling rate than $r$ does not have $\Theta(\log n / \eps^2)$ survivors of $S_j$. 
Recall that in the algorithm, for each level set it finds the sub-stream with the smallest sampling rate such that there are $\Theta(\log n / \eps^2)$ coordinates in the set. Then to get an estimate of the size of the level set we re-scale by the sampling probability. It follows from Lemma~\ref{lem:lowSL} that we identify the right sampling rate. 

So, we have with probability at least $1 - \poly(\eps / \log n)$ that $\tilde{s}_j \in [(1 - \eps)s_j, (1 + \eps)s_j]$.
\end{proof}

\begin{remark}
Lemma~4.6 suggests that our current algorithm requires $(\log n / \eps)^{c + 9}$buckets in the CountSketch data structure. However, this dependence can be further improved. For example, (1) for $p < 2$, we can save an $\log n$ when summing over the different level sets in Lemma~4.3. (2) Our current algorithm estimates each level-set size to $(1 \pm \varepsilon)$ accuracy. However, if a level is “tiny” (e.g., contributes only an $\varepsilon$-fraction of total $F_p$), such precision is unnecessary: an $O(1)$ constant-factor estimate suffices. Using this idea, we can first obtain a constant-factor approximation of the size of the level set and use it to choose an appropriate sampling level. This idea has appeared in Section 5.1.3 of \cite{LNWW2025unbiased}. This refinement can reduce the final exponent by two. (3) Our current algorithm partitions the coordinates into levels in the power of $(1+\varepsilon)$. We can instead define level sets using powers of $2$ and estimate the average contribution within each such level. This further reduces the exponent by one.
\end{remark}

\begin{lemma} \label{lem:nonS}
Consider some $S_j$ which does not contribute where $\zeta(1+\eps)^{t-j} \geq |a_k|$. With probability at least $1-\poly(\eps / \log n)$, $\tilde{s}_j \in [0, (1+\eps)s_j]$. 
\end{lemma}
\begin{proof}
Consider some level set $S_j$ which does not contribute. 
In this case, the algorithm may not find enough number of survivors in any subsampling stream. So, our lower bound for the size estimate is $0$. If there are enough survivors, then by the same argument as Lemma~\ref{lem:sampling_rate}, 
$\tilde{s_j}$ will be a $(1 \pm \eps)$-approximation of $s_j$. Therefore, for any level set that does not contribute, the size estimate can be an under-estimate but never an estimate by more than a $(1 + \eps)$ factor with probability at least $1-\poly(\eps/\log n)$. 
\end{proof}

\begin{lemma} \label{lem:nonS2}
    The $F_p$ norm of all coordinates in non-contributing level sets is at most $O(\eps) \cdot (\sum_{i=1}^k |a_i|^p)$. 
\end{lemma}
\begin{proof}
    There are $\frac{\log m}{ O(\eps)}$ level sets, and therefore at most that many non-contributing set. The $F_p$ norm of all the coordinates in each non-contributing set by definition is at most $\frac{\eps^2}{\log^2 n}\left(\sum_{i=1}^k |a_i|^p \right)$. 
     Therefore, the $F_p$ norm of all coordinates in non-contributing level sets is at most $\frac{\log m}{O(\eps)} \cdot \frac{\eps^2}{\log^2 n} \left(\sum_{i=1}^k a_i^p \right).$
\end{proof}

Our full algorithm is given in \hyperref[alg:topk]{Ky-Fan-$k$-Norm-Estimator} (Algorithm~\ref{alg:topk}). We note that in Line~\ref{line:8}, this is $\Lambda_0$ and not $\Lambda_j$. For a contributing level which has fewer than $O(1/\eps^2)$ coordinates, we will need to find all of the coordinates in that level set as opposed to obtaining a subsample of them. This means that we need to look at the entire stream instead of the sub-stream. It is also the reason that we divide the iterations into two parts $(0, t_0)$ and $(t_0 + 1, t - 1)$. 
We are now ready to prove our \Cref{thm:top_k}.
\begin{proof}[Proof of \Cref{thm:top_k}]

        We have that with probability at least $1-\poly(\eps / \log n)$ that for a contributing level set $S_j$, we have that the estimator $\tilde{s}_j \in (1 \pm \eps) s_j$ from Lemma~\ref{lem:sampling_rate}. 
        For level sets that do not contribute, we have from Lemma~\ref{lem:nonS} that the size estimate can be an under-estimate but never an estimate by more than a $(1 + \eps)$ factor with probability at least $1-\poly(\eps/\log n)$. 

        Note that we have $t = \log m / (2\eps)$ level sets where $m = \poly(n)$ by assumption. Therefore, we can take a union bound over all the level sets for all the above events and get that they happen with constant probability. Now, we upper bound the output of our algorithm. 

        As proven above, for each level set $S_j$, the algorithm estimates $s_j$ within a $(1\pm \eps)$ factor. In addition, each coordinate value in this level is within a $(1\pm \eps)$ factor due to the range of each level set. Since our estimator takes the sum of the top $k$ values, we have that the output is at most 
        \[
        (1 + \eps)^2 \sum_{i = 1}^{k} |a_{i}|^p = \left(1 + O(\eps)\right) \sum_{i = 1}^{k} |a_{i}|^p\;.
        \]
        We now lower bound the output of our algorithm. Recall that besides the error in estimating $s_j$ for contributing sets $S_j$ and the error associated with the level set range, underestimating comes from non-contributing level sets. The sum of all the coordinates in the non-contributing level sets is at most $O(\eps) \left(\sum_{i = 1}^{k} |a_{i}|^p\right)$ by Lemma~\ref{lem:nonS2}. 
        Therefore, we have that the output of the algorithm is at least 
        \[
        (1 - \eps)^2 \sum_{i = 1}^{k} |a_{i}|^p - O(\eps) \left(\sum_{i = 1}^{k} |a_{i}|^p\right)= \left(1 - O(\eps)\right) \sum_{i = 1}^{k} |a_{i}|^p\;.\qedhere
        \]
        
\end{proof}

\section{Estimation of Trimmed-$k$ $F_p$ Moment for $p \in [0,2]$} \label{sec:trimkbody}
We now present our result for the trimmed-$k$ norm. 

\begin{theorem}[$F_p$ of trimmed-$k$ vector for $0 \leq p \leq 2$] \label{thm:trimk} 
Given $\bx \in \Z^n$, $p \in [0,2]$, $k \geq 0$, and $\eps \in (0,1)$, if we have 
    \[
    a_k^2 \ge (\eps/\log n)^c \cdot \frac{\norm{\bx_{-k}}_2^2}{k}
    \]
    for some constant $c$, then there exists a linear sketch that uses $\poly(\log n / \eps)$ bits of space (where the exponent depends on the value of $c$) and estimates $\sum_{i = k}^{n - k} |a_i|^p$ with error $\eps \left( \sum_{i = k}^{n - k} |a_i|^p  + k|a_{k - \eps k}|^p\right)$ with high constant probability. 
\end{theorem}

Our estimator \hyperref[alg:trimk]{Trimmed-$k$-Norm-Estimator} (Algorithm~\ref{alg:trimk}) is similar to that of the previous section. However, we have an additional difficulty since we need to remove the contribution of the top and bottom $k$ coordinates. We only can estimate the contribution of each level set up to a $(1\pm \eps)$-factor. Therefore, we are actually estimating $\sum_{i = u}^{u + n- 2k} a_i^p$ where $u = k \pm O(\eps) k$.

\begin{algorithm}[t]
\caption{Trimmed-$k$-Norm-Estimator ($\eps \in (0, 1]$, $p \in (0, 2]$, $k \geq 0$)}
\label{alg:trimk}
\begin{algorithmic}[1] 
\STATE Run \hyperref[alg:levelSetEst]{Level-Set-Estimator ($\eps$)}.
\STATE Find the $i_1$ such that $\sum_{j=0}^{i_1 - 1} \tilde{s}_j < k$ and $\sum_{j=0}^{i_1} \tilde{s}_j \ge k$. 
\STATE Find the $i_2$ such that $\sum_{j = 0}^{i_2 - 1} \tilde{s_j}< n - k$ and $\sum_{j=0}^{i_2} \tilde{s_j} \geq n - k$.  
\STATE top-$k$ $\leftarrow \left(\sum_{j=0}^{i_1 - 1} \tilde{s_j} \cdot \zeta^p (1+\eps)^{p(t-j)}\right) + \left(k - \sum_{j = 0}^{i_1 - 1} \tilde{s}_j\right) \zeta^p(1+\eps)^{p(t-i_1)}$. 
\STATE top-$n$-minus-$k$ $\leftarrow \left(\sum_{j=0}^{i_2 - 1} \tilde{s_j} \cdot \zeta^p (1+\eps)^{p(t-j)}\right) + \left((n-k) - \sum_{j = 0}^{i_2 - 1} \tilde{s}_j\right) \zeta^p(1+\eps)^{p(t-i_2)}$.
\STATE \textbf{Return} top-$n$-minus-$k$ $ - $ top-$k$.  
\end{algorithmic}
\end{algorithm}
At a high level our algorithm runs \hyperref[alg:levelSetEst]{Level-Set-Estimator} (Algorithm~\ref{alg:levelSetEst}) to divide the universe into level sets and estimate how many coordinates are in each level set. Then it sums up the top $n - k$ coordinates and then subtracts off the top $k$ coordinates. 
Note that the condition
~\eqref{eq:1bb} 
always holds for $a_{n-k}$ as the trimmed-$k$ vector of $\bx$ only makes sense when $k \le n /2$, which means we have $n - k = \Theta(n)$. 
We first show that $u = k \pm O(\eps) k$.  
\begin{lemma}
    $u \ge k - O(\eps) k$. 
\end{lemma}

\begin{proof}
Recall that like reasoned in the proof of \Cref{thm:top_k}, for each level set associated with the top $k$ coordinates, we either underestimate its size or estimate its size up to a $(1 \pm \eps)$-factor. Therefore, we have $u \ge k - O(\eps) k$ where equality holds when we overestimate each level set associated with the top-$k$ coordinates by a $(1 + \eps)$ factor. 
\end{proof}

\begin{lemma}
    \label{lem:non-contribute_ub}
    $u \le k + O(\eps)k $.
\end{lemma}

\begin{proof}

From Lemma~\ref{lem:sampling_rate}, we know that for each level set that contributes (by Definition~\ref{def:contribute}) we can estimate its size up to a $(1 \pm \eps)$-factor. We first bound the number of coordinates in the level sets (associated with coordinates in the top-$k$) that do not contribute and we therefore do not estimate well. 

We claim that for each level set $S_j$, if it contains at least $\Omega(k\eps^2 / \log n)$ coordinates, then the algorithm estimates its size up to a $(1\pm \eps)$ factor. To show this, consider the sub-stream with sampling rate $r = \Theta(\frac{1}{k \poly(\eps /\log n)})$. By Chernoff's bound, with high probability there will be $\Theta(1/\eps^2)$ survivors in this sub-stream. Furthermore, since we are guaranteed that     
$ a_k^2 \ge \poly(\eps/\log n) \cdot \norm{\bx_{-k}}_2^2/k$, 
    this implies the coordinates in this level set can be identified by the Heavy Hitter data structure. Following the same proof as Lemma~\ref{lem:sampling_rate}, our algorithm estimates the size of this level set up to a $(1 \pm \eps)$ factor with probability $1 - \poly(\eps / \log n)$.

    Therefore, for each level set associated with coordinates in the top-$k$, we either estimate its size up to a $(1 \pm \eps)$-factor or the level set has $o(k \eps^2 / \log n)$ coordinates. Since there are at most $\log n / O(\eps)$ such level sets, we have $u \le  k + O(\eps) k$.
\end{proof}

The above discussion shows that if we can estimate $\sum_{i = u}^{u + n- 2k} |a_i|^p$ up to error 
$
\eps\left(\sum_{i = u}^{u + n- 2k} |a_i|^p\right) + \eps \cdot k |a_k|^p,
$
then the overall error of our estimator is 
\[O(\eps) \left(\sum_{i = k}^{n - k} |a_i|^p + k |a_{k - \eps k}|^p\right).\]
To achieve this, we shall consider a similar level-set argument as that in \Cref{sec:topk}.

We first define a ``contributing" level set. Recall Definition~\ref{def:contribute} for the definition of $S_j$ and $s_j$.
\begin{definition}
\label{def:congtibute_trimmed}
    We say the level set $S_j$ ``contributes" if 
    \[
    \sum_{i \in S_j} |x_i|^p \ge \frac{\eps^2}{\log^2 n} \left(\sum_{i = k + 1}^{n-k} |a_i|^p + k |a_k|^p\right).
    \]
\end{definition}
\begin{lemma} \label{lem:contLS}
    If 
    $S_j$ contributes and we have 
    $v = \zeta(1+\eps)^{t - j} \ge |a_{n - k}|$ then we have $v^2 \cdot s_j \ge (\eps / \log n)^{5} \cdot \norm{\bx_{-\mathrm{rank}(v)}}_2^2$
    where $\mathrm{rank}(v)$ is the rank of $v$ in array $\bx$ (i.e., the number of entries in $\bx$ with value greater than $v$) .
\end{lemma}

\begin{proof}
Since $S_j$ contributes, we have that 
\[
s_j v^p \geq (\eps / \log n)^2 \cdot \left(\sum_{i = k + 1}^{n - k} |a_i|^p + k |a_k|^p\right) \geq (\eps / \log n)^2 \cdot (n-k) |a_{n-k}|^p. 
\]
Multiplying by $v^{2-p}$ gives 
\begin{equation} \label{eq:new2}
v^2 s_j \geq (\eps / \log n)^2 \cdot (n-k)|a_{n-k}|^p v^{2-p} \geq (\eps / \log n)^2 \cdot (n-k)a^2_{n-k}
\end{equation}
where the last inequality is because we have $v \geq |a_{n-k}|$ and $p \leq 2$. Note that we have $n-k \geq k$ since the problem is only properly defined if we have $k \leq \frac{n}{2}$. Therefore, we have \[(n-k) a^2_{n-k} \geq k a^2_{n-k} \geq \|\bx_{-(n-k)}\|_2^2.\] Plugging this into (\ref{eq:new2}) gives us 
\begin{equation}\label{eq:new3}
    v^2 s_j \geq (\eps / \log n)^2 \cdot \|\bx_{-(n-k)}\|_2^2. 
\end{equation}

We next measure the difference between $\norm{\bx_{-(n - k)}}_2^2$ and $\norm{\bx_{-\mathrm{rank}(v)}}_2^2$. For the level set 
\[
T_\ell = \{i \in [n]: |x_i| \in [v/2^{\ell + 1}, v/2^\ell]\} , 
\]
where $v/2^{\ell + 1} \ge a_{n - k}$ and 
\[
T_q = \{i \in [n]: |x_i| \in [|a_{n-k}|, v / 2^q)\}. 
\]
Let $t_\ell = |T_\ell|$, $t_q = |T_q|$. We now claim that for $w \in [0,q]$ we have $t_w \le O(\log^2 n / \eps^2) \cdot s_j \cdot 2^{wp}$. Otherwise, we would have  
\[
\sum_{i = k + 1}^{n - k} |a_i|^p + k |a_k|^p\ge \left(\frac{v}{2^{w + 1}}\right)^{p} \cdot t_w \ge
\left(\frac{v}{2^{w + 1}}\right)^{p} \cdot O\left(\frac{\log^2n}{\eps^2}\right) s_j 2^{wp} \ge O\left(\frac{\log^2n}{\eps^2}\right) \left(\sum_{i \in S_j} |x_i|^p\right) \;
\]
which contradicts the fact that $S_j$ contributes. Now, taking a sum we get 
\begin{equation} \label{eq:rankNK}
\sum_{\mathrm{rank}(v) \le i \le n - k} a_i^2 \le \sum_{i \in [q]} t_i \left(\frac{v}{2^{i}}\right)^2 \leq \sum_{i \in [q]}\left(\frac{v}{2^i} \right)^2 O\left(\frac{\log^2n}{\eps^2}\right) s_j 2^{ip} \le v^2 s_j \cdot O\left(\frac{\log^3 n}{\eps^2}\right) \;.
\end{equation}
This implies that 
\[
\norm{\bx_{-(n - k)}}_2^2 \ge \norm{\bx_{-\mathrm{rank}(v)}}_2^2 - v^2 s_j \cdot O\left(\frac{\log^3 n}{\eps^2}\right) \;.
\]
Combining~\eqref{eq:new3} and~\eqref{eq:rankNK} we get immediately that
\[
v^2 \cdot s_j \ge (\eps / \log n)^{5} \cdot \norm{\bx_{-\mathrm{rank}(v)}}_2^2 \;. \qedhere
\]
\end{proof}
\begin{lemma}
\label{lem:trimmed_level_set}
 Suppose that 
 $S_j$ contributes 
where $\zeta(1+\eps)^{t - j} \ge |a_{n - k}|$. Then with probability at least $1-\poly(\eps / \log n)$ we have $\tilde{s}_j = (1 \pm \eps)s_j$.
\end{lemma}

\begin{proof}
    The proof is similar to that of Lemma~\ref{lem:sampling_rate}. Again consider some $S_j$ with value range $[v, (1+\eps)v]$ where $\zeta(1+\eps)^{t-j} \geq |a_{n-k}|$. Consider the sub-stream $\mathcal{P}$ with corresponding sampling rate $r = \min\left(1, \frac{C \log n}{s_j \eps^2}\right)$ for a sufficiently large constant $C$ and assume that there are $z$ survivors of $S_j$ in $\mathcal{P}$. Since we have a random threshold $\zeta$ in the boundary of the level set, we have with probability $1-\poly(\eps / \log n)$ that a $(1-\eps)$ fraction of them is in $[v(1+\poly(\eps / \log n)), (1+\eps) v (1-\poly(\eps / \log n))]$. 

    We next analyze the tail error of the heavy hitter data structure. Since $S_j$ contributes, we have 
    \[s_j \ge (\eps^2 / \log^2 n) \cdot \rank(v).\]

    Recall again that we use $(\log n / \eps)^{c+9}$ buckets in the HeavyHitter data structure. 
    So, by the same logic as in Lemma~\ref{lem:sampling_rate}, with probability $1-\poly(\eps / \log n)$ 
    the tail error of the HeavyHitter structure is at most $\frac{1}{\sqrt{s_j}} (\eps / \log n)^{c/2 + 9/2} \norm{\bx_{-\rank(v)}}_2$. 

    We have $v^2 \cdot s_j \geq (\eps / \log n)^{5} \cdot \norm{\bx_{-\rank(v)}}_2^2$ by Lemma~\ref{lem:contLS}. The rest of the proof goes through exactly as in Lemma~\ref{lem:sampling_rate}.  
\end{proof}

\begin{lemma} \label{lem:nonT}
    Consider some $S_j$ which does not contribute. With probability at least $1-\poly(\eps / \log n)$, $\tilde{s}_j \in [0, (1+\eps) s_j]$. 
\end{lemma}
\begin{proof}
    This follows from the same proof as Lemma~\ref{lem:nonS}. 
\end{proof}
\begin{lemma} \label{lem:nonT2}
    The $F_p$ norm of all coordinates in non-contributing level sets is at most $O(\eps) \cdot (\sum_{i=k+1}^{n-k}|a_i|^p + k|a_k|^p)$. 
\end{lemma}
\begin{proof}
    There are $\frac{\log m}{O(\eps)}$ level sets, and therefore at most that many non-contributing sets. Therefore by definition, we have that the $F_p$ norm of all coordinates in non-contributing level sets is at most 
    \[
    \frac{\log m}{O(\eps)} \cdot \frac{\eps^2}{\log^2 n} \left( \sum_{i = k+1}^{n-k} |a_i|^p + k|a_k|^p \right) = O(\eps) \cdot \left( \sum_{i = k+1}^{n-k} |a_i|^p + k|a_k|^p \right). \qedhere
    \]
\end{proof}

After obtaining Lemma~\ref{lem:trimmed_level_set}, similarly to the proof of \Cref{thm:top_k}, we can get the correctness of \Cref{thm:trimk}.
\begin{proof}[Proof of \Cref{thm:trimk}]
We have that with probability at least $1-\poly(\eps / \log n)$ that for contributing level set $S_j$, we have that the estimator $\tilde{s}_j \in (1\pm \eps) s_j$ from Lemma~\ref{lem:trimmed_level_set}. For level sets that do not contribute, we have from Lemma~\ref{lem:nonT}
that the size estimate can be an under-estimate but never an estimate by more than a $(1 + \eps)$ factor with probability at least $1-\poly(\eps/\log n)$. 

Note that we have $t = \log m / (2\eps)$ level sets where $m = \poly(n)$ by assumption. Therefore, we can take a union bound over all the level sets for all the above events and get that they happen with constant probability. Now, we upper bound the output of our algorithm. 

As proven above, for each contributing level set $S_j$, the algorithm estimates $s_j$ within a $(1\pm \eps)$ factor. In addition, each coordinate value in this level is within a $(1\pm \eps)$ factor due to the range of each level set. So, the output is at most  
 \[  
(1 + \eps)^2 \sum_{i = u}^{u + n-2k} |a_{i}|^p = \left(1 + O(\eps)\right) \sum_{i = u}^{u+n - 2k} |a_{i}|^p\;.
\]
        We now lower bound the output of our algorithm. Recall that besides the error in estimating $s_j$ for contributing sets $S_j$ and the error associated with the level set range, underestimating comes from non-contributing level sets. The sum of all the coordinates in the non-contributing level sets is at most $O(\eps) \left(\sum_{i = k+1}^{n-k} |a_{i}|^p + k|a_k|^p\right) $ by Lemma~\ref{lem:nonT2}. 
        Therefore, we have that the output of the algorithm is at least 
        \[
        (1 - \eps)^2 \sum_{i = u}^{u+n-2k} |a_{i}|^p - O(\eps) \left(\sum_{i = k+1}^{n-k} |a_{i}|^p + k|a_k|^p\right)= \left(1 - O(\eps)\right) \left(\sum_{i = k+1}^{n-k} |a_{i}|^p + k|a_{k-\eps k}|^p\right)\;. \qedhere
        \]

\end{proof}

\section{Trimmed Statistics for $p > 2$} \label{sec:p2}

In this section, we give our sketching algorithms for the trimmed statistic of a vector for the case when $p > 2$. We use the same algorithms as \Cref{alg:topk} and \Cref{alg:trimk} but instead keep track of the $\poly(\eps / \log n) \cdot \frac{1}{n^{1 - 2/p}}$- $\ell_2$ heavy hitters. 

\subsection{Top-$k$}
\begin{theorem}[$F_p$ of top-$k$ vector for $p > 2$] \label{thm:top_k_lp}
Given $\bx \in \Z^n$, $p > 2$, $k \geq 0$, and $\eps \in (0,1)$, suppose that we have $|a_k|^p \ge (\eps/\log n)^c \cdot \frac{\| \bx_{-k}\|_p^p}{k}$
    for some constant $c$. Then there exists a linear sketch that uses $\poly(\log n / \eps)\cdot n^{1 - 2/p}$ bits of space and estimates $\sum_{i = 1}^k |a_i|^p$ up to a $(1 \pm \eps)$ multiplicative factor with high constant probability.
\end{theorem}

We first consider the estimation of the $F_p$ moment of the top-$k$ frequencies. We use the same definition of ``contribute" as Definition~\ref{def:contribute}. 

The following lemma is analogous to Lemma~\ref{lem:condition} for the case when $p \le 2$.

\begin{lemma} \label{lem:conGP}
     For any $j$ such that that $\zeta (1+\eps)^{t-j} \ge |a_k|$ and $S_j$``contributes", 
    taking $v = \zeta(1+\eps)^{t - j - 1}$ we have $v^p \cdot s_j \ge (\eps / \log n)^{c+5} \cdot \norm{\bx_{-\mathrm{rank}(v)}}_p^p$
    where $\mathrm{rank}(v)$ is the rank of $v$ in array $\bx$ (i.e., the number of entries in $\bx$ with value greater than $v$) .
\end{lemma}
\begin{proof}
Since $S_j$ contributes by definition we have 
\begin{equation}
\label{eq:p1}
v^p \cdot s_j \ge (\eps / \log n)^{2} \cdot |a_k|^p \cdot k \ge (\eps/\log n)^{c+2} \norm{\bx_{-k}}_p^p \;
\end{equation}
where the last inequality comes from the condition that $a_k^p \geq (\eps/\log n)^c \cdot \|\bx_{-k}\|_p^p / k$. 

We next measure the difference between $\norm{\bx_{-k}}_p^p$ and $\norm{\bx_{-\mathrm{rank}(v)}}_p^p$. Define the level set $T_\ell$ for $\ell \in [0,q-1]$ where 
\[
T_\ell = \{i \in [n]: |x_i| \in [v/2^{\ell + 1}, v/2^\ell)\} , 
\]
for $v/2^{\ell + 1} \ge |a_k|$ and 
\[
T_q = \{i\in [n]: |x_i| \in [|a_k|, v/2^q)\} \;. 
\]
Let $t_\ell = |T_\ell|$, $t_q = |T_q|$. Now, we claim that  we have for $w \in [0, q]$ that $t_w \le O(\log^2n / \eps^2) \cdot s_j \cdot 2^{wp}$. Otherwise we would have 
\[
\left(\sum_{i = 1}^k |a_i|^p\right) 
\ge \left(\frac{v}{2^{w + 1}}\right)^{p} \cdot t_w
\ge \left(\frac{v}{2^{w + 1}}\right)^{p} \cdot O\left(\frac{\log^2n}{\eps^2}\right) s_j 2^{wp} \ge O \left(\frac{\log^2n}{\eps^2} \right)
\left(\sum_{i \in S_j} |x_i|^p\right) \;
\]
which contradicts the fact that $S_j$ contributes. So now, taking a sum we get 
\[
\sum_{\mathrm{rank}(v) \le i \le k} |a_i|^p \le \sum_{i \in [q]} t_i \left(\frac{v}{2^{i}}\right)^p  \leq O\left(\frac{\log^2n}{\eps^2}\right) s_j \sum_{i \in [q]} \left( \frac{v}{2^i} \right)^p 2^{ip} \le v^p s_j \cdot O\left(\frac{\log^3 n}{\eps^2}\right) \;.
\]
So we get
\begin{equation} \label{eq:rank1}
\norm{\bx_{-k}}_p^p \ge \norm{\bx_{-\mathrm{rank}(v)}}_p^p - v^p s_j \cdot O\left(\frac{\log^3 n}{\eps^2}\right) \;.
\end{equation}
Combining~\eqref{eq:p1} and~\eqref{eq:rank1} we immediately get that
\[
v^p \cdot s_j \ge (\eps / \log n)^{c+5} \cdot \norm{\bx_{-\mathrm{rank}(v)}}_p^p \;. \qedhere
\]
\end{proof}

\begin{lemma}
    \label{lem:sampling_rate1}
    Suppose that $|a_k|^p \ge (\eps/\log n)^c \cdot \frac{\norm{\bx_{-k}}_p^p}{k}$ and $S_j$ contributes where $\zeta (1+\eps)^{t-j} \ge |a_k|$. Then, with probability at least $1 - \poly(\eps / \log n)$, we have $\tilde{s}_j \in [(1 - \eps)s_j, (1 + \eps)s_j]$.
\end{lemma}

\begin{proof}
    Consider a level set $S_j$ with a value range in $[v, (1 + \eps) v]$ where $\zeta (1+\eps)^{t-j} \ge |a_k|$. Consider the sub-stream $\mathcal{P}$ with corresponding sampling rate $r  = \min\left(1, \frac{C \log n}{s_j \eps^2}\right)$ for a sufficiently large constant $C$ and assume there are $z$ survivors of $S_j$ in $\mathcal{P}$.
    Since we have a random threshold $\zeta$ in the boundary of the level set, we have with probability at least $1 - \poly(\eps /\log n)$, a $(1 - \eps)$ fraction of them is within $[v\left(1 + \poly(\eps / \log n)\right), 
(1 + \eps) v \left(1 - \poly(\eps / \log n)\right)]$. 

    We next analyze the tail error of the heavy hitter data structure. First we have that 
    \begin{equation} \label{eq:sjRank1}
    s_j \ge (\eps^2 / \log^2 n) \cdot \rank(v).
    \end{equation}
    Suppose that we had $s_j < (\eps^2 / \log^2 n) \cdot \rank(v)$. This would mean that $\rank(v) \cdot v^p \geq s_j (\log^2 n / \eps^2) \cdot v^p$, which contradicts the fact that $S_j$ 
    contributes. Therefore, combining 
    (\ref{eq:sjRank1}) with the fact that  
    $\mathcal{P}$ 
    has sampling rate 
    $r = \min\left(1, \frac{C \log n}{s_j \eps^2}\right)$ gives us that with probability $1- \poly(\eps / \log n)$ using Chernoff's bound that the number of survivors of the top $\rank(v)$ coordinates of $x$ surviving in $\mathcal{P}$ is at most $O((\log n / \eps)^4)$. 
    
    Suppose that we use $(\log n/\eps)^{c+9} \cdot n^{1-2/p}$ buckets in our heavy hitter data structure (\Cref{sec:count-sketch}). Hence, with probability at least $1 - \poly(\eps /\log n)$, the tail error of the heavy hitter data structure will be at most $\frac{1}{\sqrt{s_j}} (\eps / \log n)^{c/2+9/2} \cdot \norm{\bx_{-\rank(v)}}_p$ (converting from $\| \cdot \|_2$ to $\| \cdot \|_p$ by Holder's Inequality).  
    
    Recall that we have condition $|a_k|^p \ge (\eps/\log n)^c \cdot \frac{\norm{\bx_{-k}}_p^p}{k}$ and therefore from Lemma~\ref{lem:conGP} have $v^p \cdot s_j \ge (\eps / \log n)^{c+5} \cdot \norm{\bx_{-\mathrm{rank}(v)}}_p^p$. 
     Combining this with the above tail error we immediately have with high probability the heavy hitter data structure can identify every survivor in $[v\left(1 + \poly(\eps / \log n)\right), 
(1 + \eps) v \left(1 - \poly(\eps / \log n)\right)]$. 

Combining this with Lemma~\ref{lem:sjapprox}, the remaining thing is to show that with high probability that a level with a smaller sampling rate than $r$ does not have $\Theta(\log n / \eps^2)$ survivors of $S_j$. 
Recall that in the algorithm, for each level set it finds the sub-stream with the smallest sampling rate such that there are $\Theta(\log n / \eps^2)$ coordinates in the set. Then to get an estimate of the size of the level set we re-scale by the sampling probability. It follows from Lemma~\ref{lem:lowSL} that we identify the right sampling rate. 

So, we have with probability at least $1 - \poly(\eps / \log n)$ that $\tilde{s}_j \in [(1 - \eps)s_j, (1 + \eps)s_j]$.
\end{proof}

The rest of the proof of~\Cref{thm:top_k_lp} follows from the proof of \Cref{thm:top_k} and from the fact that Holder's inequality gives us $\norm{\cdot}_2 \cdot \frac{1}{n^{1/2 - 1/p}} \leq \norm{\cdot}_p$ for $p > 2$.

\subsection{Trimmed-$k$}

\begin{theorem}[$F_p$ for trimmed-$k$ vector for $p > 2$] \label{thm:trimmed_k_lp}
Given $\bx \in \Z^n$, $p > 2$, $k \geq 0$, and $\eps \in (0,1)$, suppose that we have $|a_k|^p \ge (\eps/\log n)^c \cdot \frac{\| \bx_{-k}\|_p^p}{k}$
    for some constant $c$. Then there exists a linear sketch that uses $\poly(\log n / \eps)\cdot n^{1 - 2/p}$ bits of space and estimates $\sum_{i = k + 1}^{n-k} |a_i|^p$ with error $\eps \left(\sum_{i=k + 1}^{n-k} |a_i|^p + k |a_{k - \eps k}|^p \right)$ with high constant probability. 
\end{theorem}

We next consider the estimation of the $F_p$ moment of the trimmed-$k$ vector. We say a level set $S_j$ ``contributes" according to 
Definition~\ref{def:congtibute_trimmed}. 
The following lemma is analogous to Lemma~\ref{lem:contLS} for the case when $p \le 2$.

\begin{lemma} \label{lem:contLS_lp}
    If 
    $S_j$ contributes and we have 
    $v = \zeta(1+\eps)^{t - j} \ge |a_{n - k}|$ then we have $v^p \cdot s_j \ge (\eps / \log n)^{c+5} \cdot \norm{\bx_{-\mathrm{rank}(v)}}_p^p$
    where $\mathrm{rank}(v)$ is the rank of $v$ in array $\bx$ (i.e., the number of entries in $\bx$ with value greater than $v$) .
\end{lemma}
\begin{proof}
Since $S_j$  contributes, by definition, we have 
\begin{equation} \label{eq:333}
    v^p \cdot s_j \geq (\eps / \log n)^{2} \cdot a_{n-k}^p \cdot (n-k) \geq O(\eps / \log n)^{2} \norm{\bx_{-(n - k)}}_p^p
\end{equation}
given $n - k  \ge \frac{n}{2}$. Note that the problem is only properly defined if we have $k \leq \frac{n}{2}$.   

We next measure the difference between $\norm{\bx_{-(n - k)}}_p^p$ and $\norm{\bx_{-\mathrm{rank}(v)}}_p^p$. For the level set 
\[
T_\ell = \{i \in [n]: |x_i| \in [v/2^{\ell + 1}, v/2^\ell]\} , 
\]
where $v/2^{\ell + 1} \ge |a_{n - k}|$ and 
\[
T_q = \{i \in [n]: |x_i| \in [|a_{n-k}|, v / 2^q)\}. 
\]
Let $t_\ell = |T_\ell|$, $t_q = |T_q|$. We now claim that for $w \in [0,q]$ we have $t_w \le O(\log^2 n / \eps^2) \cdot s_j \cdot 2^{wp}$. Otherwise, we would have  
\[
\sum_{i = k + 1}^{n - k} |a_i|^p + k|a_k|^p \ge \left(\frac{v}{2^{w + 1}}\right)^{p} \cdot t_w \ge
\left(\frac{v}{2^{w + 1}}\right)^{p} \cdot O\left(\frac{\log^2n}{\eps^2}\right) s_j 2^{wp} \ge O\left(\frac{\log^2n}{\eps^2}\right) \left(\sum_{i \in S_j} |x_i|^p\right) \;
\]
which contradicts the fact that $S_j$ contributes. Now, taking a sum we get 
\begin{equation} \label{eq:rankNK1}
\sum_{\mathrm{rank}(v) \le i \le n - k} |a_i|^p \le \sum_{i \in [q]} t_i \left(\frac{v}{2^{i}}\right)^p  \leq O\left(\frac{\log^2n}{\eps^2}\right) s_j \sum_{i \in [q]} \left( \frac{v}{2^i} \right)^p 2^{ip} \le v^p s_j \cdot O\left(\frac{\log^3 n}{\eps^2}\right) \;.
\end{equation}
This implies that 
\[
\norm{\bx_{-(n - k)}}_p^p \ge \norm{\bx_{-\mathrm{rank}(v)}}_p^p - v^p s_j \cdot O\left(\frac{\log^3 n}{\eps^2}\right) \;.
\]
Combining~\eqref{eq:333} and~\eqref{eq:rankNK1} we get immediately that $v^p \cdot s_j \ge (\eps / \log n)^{5} \cdot \norm{\bx_{-\mathrm{rank}(v)}}_p^p.$

\end{proof}

The rest of the proof of~\Cref{thm:trimmed_k_lp} follows from the proof of \Cref{thm:trimk} and from the fact that Holder's inequality gives us $\norm{\cdot}_2 \cdot \frac{1}{n^{1/2 - 1/p}} \leq \norm{\cdot}_p$ for $p > 2$.

\section{Applications} \label{sec:apps}
In this section, we will consider some variants of the problem we have considered thus far. 

\subsection{Sum of Large Items}

\begin{corollary}[Thresholded $F_p$ estimation] \label{thm:sumHH} 
Given $\bx \in \Z^n$, $ p \in [0,2]$, threshold $\mathcal{T} \geq 0$, and $\eps \in (0,1)$, if we have $a_k^2 \geq (\eps / \log n)^c \cdot \frac{\|\bx_{-k}\|_2^2}{k},$ 
    for some constant $c$ where $k$ is the largest integer such that $|a_k| \geq \mathcal{T}$, then there exists a linear sketch that uses $\poly(\log n / \eps)$ bits of space and estimates $\sum_{i \in \mathcal{B}_{\T}}|x_i|^p$ for $\mathcal{B}_\mathcal{T} = \{ i \in n : |x_i| \geq \mathcal{T} \}$ with error $\eps \left (\sum_{i \in \mathcal{B}_{\mathcal{T}}} |x_i|^p \right) + (1+\eps) \mathcal{T}^p \cdot |x_{(1-\eps)\mathcal{T}, \mathcal{T}}|$ with high constant probability where $|x_{(1-\eps)\mathcal{T}, \mathcal{T}}|$ denotes the number of coordinates with value $[(1-\eps)\mathcal{T}, \mathcal{T})$. 
\end{corollary}

The first application is the estimation of the $F_p$ moment of the frequencies that are larger than a given threshold $\mathcal{T}$ in absolute value. At a high level our algorithm, \hyperref[alg:summedHH]{Summed-HH-Estimator} (Algorithm~\ref{alg:summedHH}), divides the items in level sets. Then we take the sum over all level sets associated with a value that is above the threshold. We note that there is extra difficulty here since level sets only allow us to estimate coordinates within a $(1 + \eps)$-factor. In particular, we cannot distinguish between coordinates with value $(1-\eps) \mathcal{T}$ and $\mathcal{T}$ causing us to incur additional additive error. 

\begin{algorithm}[t]
\caption{Summed-HH-Estimator ($\eps \in (0, 1]$, $p \in (0, 2]$, $k \geq 0$, $\mathcal{T} \geq 0$)}
\label{alg:summedHH}
\begin{algorithmic}[1] 
\STATE Run \hyperref[alg:levelSetEst]{Level-Set-Estimator ($\eps$)}.
\STATE $t \leftarrow \log_{1+\eps}(m) + 1$.
\STATE Find the largest $i$ such that $\zeta (1+\eps)^{t-i} \geq \mathcal{T}$. 
\STATE \textbf{Return} $\left(\sum_{j=0}^{i} \tilde{s_j} \cdot \zeta^p (1+\eps)^{p(t-j)}\right)$. 
\end{algorithmic}
\end{algorithm}

\begin{proof} [Proof of~\Cref{thm:sumHH}]
For the purposes of analysis, take $k$ to be the largest integer such that $|a_k| \geq \mathcal{T}$. 
    Let $i$ be the largest integer number such that $\zeta (1+\eps)^{t-i} \geq \mathcal{T}$ and let $k'$ be the largest integer number such that $|a_k'| \ge \zeta (1+\eps)^{t-i - 1}$. We can see the output $Z$ of \Cref{alg:summedHH} will be the same as the output of \Cref{alg:topk} with input parameter $k'$.
    From the assumption we have that $a_{k'}^2 \geq \poly(\eps / \log n) \cdot \frac{\|\bx_{-k'}\|_2^2}{k'}, $
    as $|a_{k'}|$ can only differ from $|a_k|$ by a $(1 \pm \eps)$-factor.
    From \Cref{thm:top_k} we have that with high constant probability $\left|Z - \sum_{i = i}^{k'} |a_i|^p\right| \le \eps \left(\sum_{i = 1}^{k'} |a_i|^p \right)$

    We next measure the difference between $\left(\sum_{i = 1}^{k'} |a_i|^p \right)$ and $\left(\sum_{i = 1}^{k} |a_i|^p \right)$. Clearly, it can be upper bounded by $\mathcal{T}^p \cdot |x_{(1-\eps)\mathcal{T}, \mathcal{T}}|$ where $|\bx_{(1-\eps)\mathcal{T}, \mathcal{T}}|$ denotes the number of coordinates with value $[(1-\eps)\mathcal{T}, \mathcal{T}) $. Thus by triangle inequality we get that with high constant probability, $Z$ is an estimation of $\left(\sum_{i = 1}^{k} |a_i|^p \right)$ with error at most $\eps \left (\sum_{i \in \mathcal{B}_{\mathcal{T}}} |x_i|^p \right) + (1+\eps) \mathcal{T}^p \cdot |\bx_{(1-\eps)\mathcal{T}, \mathcal{T}}|$.
\end{proof}

\subsection{Extensions of Impact Indices}
\subsubsection{Extension of the $g$-index} \label{sec:thresk}

\begin{corollary} \label{thm:sumAboveT}
     Given $\bx \in \Z^n$, $p \in [1,2]$ and $\eps \in (0,1)$, let $k$ be the largest integer number such that $\sum_{i=0}^k |a_i|^p \geq k^{p+1}$. If we have 
    $a_{(1 - \eps )k}^2 \geq (\eps / \log n)^c \cdot \frac{\norm{\bx_{-(1 - \eps)k}}_2^2}{(1 - \eps) k}$
    for some constant $c$, then there exists a linear sketch that uses $\poly(\log n/\eps)$ bits of space and outputs a $\tilde{k}$ such that $(1 - \eps) \cdot k \le \tilde{k} \le (1 + \eps) \cdot (k + 1)$ with high constant probability. 
\end{corollary}

Here we prove \Cref{thm:sumAboveT} with our algorithm \hyperref[alg:thresAlg]{Threshold-Norm-Estimator} (\Cref{alg:thresAlg}). We run \hyperref[alg:levelSetEst]{Level-Set-Estimator} and then use that to find our estimate to $k$. 

\begin{algorithm}[t]
\caption{Threshold-Norm-Estimator ($\eps \in (0, 1]$, $p \in (0, 2]$)}
\label{alg:thresAlg}
\begin{algorithmic}[1] 
\STATE Run \hyperref[alg:levelSetEst]{Level-Set-Estimator ($\eps$)}.
\STATE Find the largest $\tilde{k}$ such that the estimated $F_p$ moment of the top-$k'$ frequencies by \Cref{alg:topk} $\geq \tilde{k}^{p+1}$.
\STATE \textbf{Return} $\tilde{k}$.
\end{algorithmic}
\end{algorithm}

\begin{proof} [Proof of~\Cref{thm:sumAboveT}]

Take $k$ to be the largest integer such that $\sum_{i=1}^k |a_i|^p \geq k^{p+1}$. We show that the estimate of the algorithm $\tilde{k}$ incurs an appropriate error. In the rest of the proof, we assume the estimation of each level set in \Cref{alg:levelSetEst} satisfies the guarantee, which holds with high constant probability. We first lower bound $\tilde{k}$.
\begin{lemma} \label{lem:k1}
    $\tilde{k} \geq (1-\eps) \cdot k$. 
\end{lemma}
\begin{proof}
It is sufficient to show that the algorithm will find $(1-\eps) \cdot k$ frequencies such that their \emph{estimated} $F_p$ moment is at least $(1-\eps)^{p+1} \cdot k^{p+1}$. 

By definition, the top $k$ frequencies have a $F_p$ norm of at least $k^{p+1}$. So, the top $(1-\eps) \cdot k$ frequencies have a $F_p$ norm of at least $(1-\eps) \cdot k^{p+1}$.  
Since we have condition $ a_{(1 - \eps )k}^2 \geq \poly(\eps / \log n) \cdot \frac{\norm{\bx_{-(1 - \eps)k}}_2^2}{(1 - \eps) k}$, by \Cref{thm:top_k}, estimating the $F_p$ moment of these top $(1-\eps) \cdot k$ frequencies incurs at most $\eps \cdot \sum_{i=1}^{(1-\eps) \cdot k} |a_i|^p$ error. So, we have that the estimated $F_p$ norm of the top $(1-\eps) \cdot k$ frequencies is at least \[(1-\eps) \cdot (1-\eps) \cdot k^{p+1} \geq (1-\eps)^{p+1} \cdot k^{p+1}.\] Recall that we have $p \in [1,2]$. 
\end{proof}
Now we upper bound $\tilde{k}$. We first prove the following which is needed to upper bound $\tilde{k}$. 
\begin{lemma} \label{lem:noCon}
    If the condition $ a_j^2 \geq \poly(\eps / \log n) \cdot \frac{\|\bx_{-j}\|_2^2}{j}$ is not met, with high probability \Cref{alg:topk} only overestimates $\sum_{i=1}^j |a_i|^p$ by a $(1+\eps)$ multiplicative factor. 
\end{lemma}
\begin{proof}
    Recall that in the proof of \Cref{thm:top_k}, we have the property that with high probability the estimation of $s_i$ for each level set is either a $(1 + \eps)$-approximation (when the condition for $a_j$ is met and the level set "contribute") or only an over-estimation by at most a $(1 + \eps)$-factor. This means that the output of the algorithm is at most
        \[ 
        (1 + \eps)^2 \sum_{i = 1}^{j} |a_{i}|^p = \left(1 + O(\eps)\right) \sum_{i = 1}^{j} |a_{i}|^p\;. \qedhere
        \]
\end{proof}
\begin{lemma} \label{lem:k2}
    $\tilde{k} \leq (1+\eps) \cdot (k+1).$
\end{lemma}
\begin{proof}
The algorithm finds the largest $\tilde{k}$ such that the \emph{estimated} $F_p$ norm of the top $\tilde{k}$ frequencies is at least $\tilde{k}^{p+1}$. Therefore, we will argue that for every $k' > (1+\eps) (k+1)$, the estimated $F_p$ norm of the top $k'$ frequencies is (strictly) less than $(k')^{p+1}$. 

By definition, the top $(k + 1)$ frequencies have a $F_p$ norm strictly less than $(k+1)^{p+1}$. Otherwise, this contradicts that $k$ is the correct answer. Therefore, the $F_p$ norm of the top $c \cdot (k+1)$ frequencies for some $c$ is at most $c \cdot (k+1)^{p+1}$. 

From Lemma~\ref{lem:noCon} we know that in either case $\sum_{i=1}^{j} |a_i|^p$ is overestimated by at most a $(1+\eps)$ factor for $j = c\cdot(k + 1)$. So we have that the estimated $F_p$ norm of the top $c \cdot (k+1)$ frequencies is less than $(1+\eps) \cdot c \cdot (k+1)^{p+1} \leq c^{p+1}\cdot k^{p+1}$ given that $c > (1+\eps)$. 
\end{proof}
Combining Lemma~\ref{lem:k1} and Lemma~\ref{lem:k2} we get the correctness of \Cref{thm:sumAboveT}. 
\end{proof}

\subsubsection{Extension of the $h$-index and $a$-index}

\begin{corollary}
    \label{thm:generalIndex}
For input $\bx \in \Z^n$, take $k$ to be the largest integer such that $|a_k| \geq k$. Given $\bx \in \Z^n$, $p \in [0,2]$ and $\eps \in (0,1)$, if we have 
$a_{(1 + \eps)k }^2 \geq (\eps / \log n)^c \cdot \frac{\norm{\bx_{-(1 + \eps )k}}_2^2}{(1 + \eps)   k}$
for some constant $c$, then there exists a linear sketch that uses $\poly(\log n / \eps)$ bits of space and estimates $\sum_{i=1}^k |a_i|^p$ up to a $(1 \pm \eps)$ multiplicative factor with high constant probability. 
\end{corollary}

 Here we prove \Cref{thm:generalIndex} with our algorithm \hyperref[alg:indexAlg]{Index-Norm-Estimator} (\Cref{alg:indexAlg}). We run \hyperref[alg:levelSetEst]{Level-Set-Estimator} and then use that to estimate the value of $k$. Then we return the $F_p$ norm of the top $k'$ frequencies. 

 \begin{algorithm}[t]
\caption{Index-Norm-Estimator ($\eps \in (0, 1]$, $p \in (0, 2]$)}
\label{alg:indexAlg}
\begin{algorithmic}[1] 
\STATE Run \hyperref[alg:levelSetEst]{Level-Set-Estimator ($\eps/10$)}.
\STATE Take $f$ to be the vector where the first $\tilde{s}_0$ elements are $ \zeta(1+\eps)^{t}$, the next $\tilde{s}_1$ elements are  $ \zeta^{(1+\eps)^{(t-1)}}$, the next $\tilde{s}_2$ elements are $ \zeta(1+\eps)^{(t-2)}$, and so on. 
\STATE Find the largest $\tilde{k}$ such that $f_{\tilde{k}} \geq \tilde{k}$.
\STATE \textbf{Return} $\sum_{i=1}^{\tilde{k}} f_i^p$.
\end{algorithmic}
\end{algorithm}

\begin{proof} [Proof of~\Cref{thm:generalIndex}]
Let $k$ denote the largest integer such that $|a_k| \geq k$. 

\begin{lemma}
    With high probability, we have $(1 - \eps / 2 ) \cdot k \le k' \le (1 + \eps/4) \cdot k$
\end{lemma}
\begin{proof}
We first show that with high probability $\tilde{k} \geq (1-\eps/2) \cdot k$. By definition, there are at least $k$ frequencies each having frequencies at least $k$. From Lemma~\ref{lem:sampling_rate} we know that with high probability for each contributing $S_j$, we have $\tilde{s}_j \geq (1-\eps/10) \cdot s_j$. 
From Lemma~\ref{lem:non-contribute_ub} we can get the total number of coordinates in a non-contribute level set is at most $\eps / 10$. Put these two things together, we can get that with high probability, the level set estimator will find at least $(1 - \eps/2)$ coordinates having frequencies at least $(1 - \eps/2) \cdot k$. This means that $k' \ge (1 - \eps / 2) \cdot k$.

We next show that with high probability $k' \le (1 + \eps / 4) \cdot k$. By definition, there are at most $k$ frequencies each having value strictly greater than $k$. From Lemma~\ref{lem:sampling_rate} and Lemma~\ref{lem:nonS} we know that with high probability for each contributing $S_j$, we have $\tilde{s}_j \le (1 + \eps/10) \cdot s_j$. 
This means that with high probability, the level set estimator can find at most $(1 + \eps / 4)$ coordinates having frequencies at most $(1 + \eps / 4) \cdot k$, which implies $k' \le (1 + \eps / 4) \cdot k$.
\end{proof}

After obtaining the lower bound and upper bound $k'$, recall that the output of \Cref{alg:indexAlg} is the estimation of $F_p$ moment the top-$k'$ frequencies. Given $(1 - \eps / 2) \cdot k \le k' \le (1 + \eps / 4) \cdot k$, from \Cref{thm:top_k} we have that with high probability, the output $Z$ of \Cref{alg:indexAlg} satisfies
\[
(1 - \eps / 2) (1 - \eps / 2) \left( \sum_{i = 1}^{k} |a_i|^p \right)\le  Z \le (1 + \eps / 4) (1 + \eps / 2) \left( \sum_{i = 1}^{k} |a_i|^p \right) \;.
\]
This implies $Z$ is an approximation to $\sum_{i = 1}^{k} |a_i|^p $ within error at most $\eps \cdot \left( \sum_{i = 1}^{k} |a_i|^p \right)$
\end{proof}

\section{Lower Bounds} \label{sec:LBs}

\subsection{Hardness of $F_p$ of Top-$k$} \label{sec:hardness}
\subsubsection{Main bound}
We show that when the condition
$a_k^2 \ge \poly(\eps/\log n) \cdot \frac{\norm{\bx_{-k}}_2^2}{k}$
does not hold, then $n^{O(1)}$ space is required. Formally, we have the following.

\begin{lemma}
\label{lem:top_k_lower_bound_new}
    Suppose that $a_k^2 \le c\cdot \frac{k}{n} \cdot \frac{\norm{\bx_{-k}}_2^2}{k}$ 
    for some constant $c$. Assume $k \le 0.1 n$ and $\eps \in (1/\sqrt{k},1]$. Then, any $O(1)$-pass streaming algorithm that outputs a $(1 + \eps)$-approximation of $\sum_{i = 1}^k |a_i|^p$ with high constant probability requires $\Omega(\eps^{-2}n/k)$ bits of space.
\end{lemma}

We will consider the following variant of the 2-SUM problem in~\cite{CLLTWZ2024tight}.

\begin{definition}
    For binary strings $\bx = (x_1,\dots, x_L) \in \{0, 1\}^L$ and $\by = (y_1, \dots, y_L) \in \{0, 1\}^L$, define $\textsc{INT}(\bx, \by) = \sum_{i=1}^L x_i \wedge y_i$ which is the number of indices $i$ where both $x_i$ and $y_i$ are $1$ and define $\textsc{DISJ}(\bx, \by)$ to be $1$ if $\textsc{INT}(\bx, \by) = 0$ and $0$ otherwise.
\end{definition}

\begin{definition}[\cite{WZ2014optimal, CLLTWZ2024tight}]
    Suppose Alice has $t$ binary strings $(\bX^1, \dots, \bX^t)$ where each string $\bX^i \in \{0, 1\}^L$ has length $L$. Likewise, Bob has $t$ strings $(\bY^1, \dots, \bY^t)$ each of length $L$. $\textsc{INT}(\bX^i, \bY^i)$ is guaranteed to be either $0$ or $\alpha \geq 1$ for each pair of strings $(\bX^i, \bY^i).$ Furthermore, at least $1/2$ of the $(\bX^i, \bY^i)$ pairs are guaranteed to satisfy $\textsc{INT}(\bX^i, \bY^i) = \alpha.$ In the $\textsc{2-SUM}(t, L, \alpha)$ problem, Alice and Bob attempt to approximate $\sum_{i \in [t]} \textsc{DISJ}(\bX^i, \bY^i)$ up to an additive error of $\sqrt{t}$ with constant probability.
\end{definition}

\begin{lemma}[\cite{CLLTWZ2024tight}] 
\label{lem:2sum}
 To solve $\textsc{2-SUM}(t, L, \alpha)$ with constant probability,
    the expected number of bits Alice and Bob need to communicate is $\Omega(tL/\alpha)$.    
\end{lemma}

We remark that the construction of the input distribution in~\cite{WZ2014optimal} for this problem has the property such that for every pair $(X^i_j,  Y^i_j)$, the probability of $(X^i_j + Y^i_j) \ge 1$ is $\Theta(1)$. Hence, we can also add the promise that there is at most a constant fraction of $(X^i_j,  Y^i_j)$ over $i, j$ having $(X^i_j + Y^i_j) = 0$. We will use this in our proof. 

We are now ready to prove our Lemma~\ref{lem:top_k_lower_bound_new}. At a high level, we will show a reduction from the estimation of $\sum_{i = 1}^{k} |a_i|^p$ to the 2-SUM problem and derive a $O(\eps^{-2} n/k)$ lower bound.

\begin{proof}[Proof of \Cref{lem:top_k_lower_bound_new}]
    Without loss of generality, we assume $\eps^2 k$ is an integer. Consider the $2$-SUM$(\eps^{-2}, \eps^{2}n, \eps^2 
k)$ problem with the input instance be $(\bX^{1}, \bX^{2}$ $, \ldots, \bX^{\eps^{-2}})$ given to Alice and $(\bY^{1}, \bY^{2}, \ldots, \bY^{\eps^{-2}})$ given to Bob. Let $\bX, \bY \in \{0, 1\}^n$ be the concatenation of the $\bX^i$'s and $\bY^i$'s respectively. We construct the input array $\bA = \bx + \by$ for our top-$k$ sum problem as follows. For $\bx, \by \in \{0, k/2\}^n$, we let $x_i = k/2$ if and only if $X_i = 1$, and $y_i = k/2$ if and only if $Y_i = 1$. 

We next consider the entry of the array $\bA = \bx + \by$. It is clear that the coordinates of $\bA$ are $0, k/2,$ or $k$. Let $c$ be the fraction of the $\bX^i, \bY^i$ such that $\textsc{INT}(\bX^i, \bY^i) = \alpha =  \eps^2 
k$. Recall that from the assumption on the input, we know that $\frac{1}{2} \le c \le 1$. We first consider the number of coordinates of $\bA$ that are equal to $k$. From the definition of $\bA$ we get that it is equal to $c \cdot \eps^{-2} \cdot \eps^2 
k = ck$. Moreover, from the assumption of the input we have all of the top $k$ entries have a value of either $k/2$ or $k$.
Let $\ba$ be an array with the decreasing order of the array $\bA = \bx + \by$. Then we have 

\[
a_k^2 < c \cdot \frac{k}{n} \cdot \frac{\norm{\bA_{-k}}_2^2}{k} \;. 
\]

The reduction is given as follows. Suppose that there is a $q = O(1)$ pass streaming algorithm $\mathcal{A}$ which gives a $(1 \pm \eps)$-approximation to the top-$k$ sum of the input array. 

\begin{enumerate}
        \item Given Alice's strings $(\bX^1, \ldots, \bX^{\eps^{-2}})$ each of length $\eps^2 n$, let $\bX$ be the concatenation of Alice's strings having total length $\eps^{-2} (\eps^2 n) = n$. Similarly let $\bY \in \{0, 1\}^n$ be the concatenation of Bob's strings.

        \item Alice constructs the vector $\bx$ as defined above and performs the updates to $\mathcal{A}$ based on $\bx$. Then Alice sends the memory of the algorithm to Bob.

        \item Bob constructs the vector $\by$ as defined above and performs the updates to $\mathcal{A}$ based on $\by$. 

        \item If the algorithm $\mathcal{A}$ is a $q \ge 2$-pass algorithm, repeat the above process $q$ times.
        \item Run $\mathcal{A}(\bx + \by)$ and compute the solution $\textsc{2-SUM}(\eps^{-2}, \eps^2 n, \eps^2 k)$ based on $\mathcal{A}(\bx + \by)$.
    \end{enumerate}

To show the correctness of the above algorithm, recall that the output of $\mathcal{A}(\bx + \by)$ is $(1 \pm \eps) \sum_{i = 1}^k |a_i|^p$ where $\ba$ is the array $\bA = \bx+ \by$ in non-increasing order. 
Note that we have shown that there are $ck$ entries among the top $k$ entries of $\bA$ having value $k$ as well as the top $k$ entries of $\bA$ having a value of either $k$ or $k/2$. This implies from the output of $\mathcal{A}(\bx + \by)$ we can get a $(1 \pm \eps)$-approximation of $c$, which yields a approximation of $\sum_{i \in [t]} \textsc{DISJ}(\bX^i, \bY^i)$ with additive error $\eps^{-1}$ (as there are a total number  of $\eps^{-2}$ of the pair $\bX^i, \bY^i$).

\end{proof}

\begin{corollary}
Suppose that $a_k^2 \le c' \cdot \frac{k}{n^c} \cdot  \frac{\norm{\bx_{-k}}_2^2}{k}$  
    for some constant $c', c \in (0, 1)$. Assume $k \le 0.1 n$ and $\eps \in (1/\sqrt{k},1]$. Then, any $O(1)$-pass streaming algorithm that outputs a $(1 \pm \eps)$-approximation of $\sum_{i = 1}^k |a_i|^p$ with high constant probability requires $\Omega(\eps^{-2} n^c/k)$ bits of space. 
\end{corollary}

\begin{proof}
    Note that we can let $m = n^c / k$ and put the above input instance in the first $n^c$ coordinates in the input to the streaming algorithm, and let the remaining coordinates to be $0$.
\end{proof}

\subsubsection{Additional bound}

We also prove the following lower bound, which shows that when $a_k$ is small enough compared to the $F_2$ moment of the tail $\norm{\bx_{-k}}_2^2$, even an $O(k^p)$ approximation is hard.

\begin{lemma}
\label{lem:top_k_lower_bound}
    Suppose that $a_k^2 \le c \cdot \frac{k^3}{n} \cdot \frac{\norm{\bx_{-k}}_2^2}{k}$
    for some constant $c$. Then, any $O(1)$-pass streaming algorithm that outputs a $O(k^p)$ approximation of $\sum_{i = 1}^k |a_i|^p$ with high constant probability requires $\Omega(n/k^3)$ bits of space.
\end{lemma}
We will consider the following $k$-player set-disjointess problem.
\begin{lemma}[\cite{G2009asymptotically,J2009hellinger, KPW2021simple}] In the $k$-players set disjointness problem, there are $k$ players with subset $S^1, S^2, \dots, S^k$, each drawn from $\{1, 2, \cdots, m\}$, and we are promised that either the sets are $(1)$ pairwise disjoint, or $(2)$ there is a unique element $j$ occuring in all the sets. To distinguish the two cases with high constant probability, the total communication is $\Omega(m/k)$ bits.
    
\end{lemma}

\begin{proof} [Proof of \Cref{lem:top_k_lower_bound}]
    We shall show that for $m = n/k$, if there is a $O(k^p)$-approximation one-pass streaming for the top-$k$ sum, then the $k$ players can use this to solve the above $k$-plyaer set disjointness problem. This immediately implies an $\Omega(m/k^2) = \Omega(n/k^3)$ lower bound (the extra $1/k$ factor here is due to the fact that the $k$ players need to send the memory status of the algorithm $k - 1$ times).

    For a player $i$, let $\ba^i \in \Z^{n/k} $ denote the binary indicator vector of $S^i$ and $\bx^i = (\ba^i, \ba^i, \dots, \ba^i) \in \Z^n$, which is formed by $k$ copies of $\ba^i$. Consider the vector $\ba = \sum_i \ba^i$, in case $(1)$, we have $\ba = (1, 1, \dots, 1)$, while in the case $(2)$, $\ba$ has one coordinate that equals to $k$ and the remaining coordinates that equals to $1$. Similarly we have the vector $\bx = \sum_i \bx^i = (1, 1, \dots, 1)$ in case $(1)$ and $\bx$ has $k$ values equals to $k$ in case $(2)$. Note that in both case we have $a_k = k$ and $\norm{\bx_{-k}}_2^2 = n - k$, which means that
    \[
    a_k^2 < \frac{k^3}{n} \cdot \frac{\norm{\bx_{-k}}_2^2}{k} \;.
    \]
    On the other hand, the sum $\sum_{i = 1}^{k} a_i^p = k$ or $k^{p + 1}$ for the two different cases. This means that use an $O(k^p)$ approximation algorithm for the top-$k$ sum problem, the $k$-players can solve the $k$-player set disjointness problem with high constant probability, which yields $\Omega(n/k^3)$ bits of space.
\end{proof}
\begin{corollary}
Suppose that $a_k^2 \le c' \cdot \frac{k^3}{n^c} \cdot \frac{\norm{\bx_{-k}}_2^2}{k}$
    for some constant $c', c \in (0, 1)$. Then, any $O(1)$-pass streaming algorithm that outputs a $O(k^p)$ approximation of $\sum_{i = 1}^k |a_i|^p$ with high constant probability requires $\Omega(n^c/k^3)$ bits of space.
\end{corollary}

\begin{proof}
    Note that we can let $m = n^c / k$ and put the above input instance in the first $n^c$ coordinates in the input to the streaming algorithm, and let the remaining coordinates to be $0$.
\end{proof}

\subsection{Hardness of $F_p$ of Trimmed-$k$}
In this section, we give a lower bound for estimating the $\norm{\bx_{-k}}_p^p$,  which shows the $\eps k \cdot |a_{k - \eps k}|^p$ error is necessary.  We consider the following Gap-Hamming problem.

\begin{definition}
    In the Gap-Hamming problem, Alice gets $\ba \in \{0, 1\}^n$ and Bob get $\bb \in \{0, 1\}^n$, and their goal is to determine the Hamming distance $\Delta(\bx, \by)$ satisfies $\Delta(\bx, \by)\ge n(c_1 - c_2 \eps)$ or $\Delta(\bx, \by)\le n(c_1 - 2c_2 \eps )$ with probability $1 - \delta$ for some constant $c_1, c_2$. 
\end{definition}

\begin{lemma}[\cite{JW2011optimal}]
    The one-way communication complexity of Gap Hamming is $\Omega(\eps^{-2}\log n \log(1/\delta))$.
\end{lemma}

\begin{lemma}
\label{lem:trim_lower_bound}

Any one-pass streaming algorithm that  
estimating $\norm{\bx_{-k}}_p^p$ with
error $\eps k\cdot |a_{k - \eps k}|^p$ and with high constant probability requires $\Omega(\eps^{-2}  \log k)$ bits of space.
\end{lemma}

\begin{proof}
    Given the binary vector $\ba, \bb \in \{0, 1\}^{2k}$, let $\bx \in \R^n = B \cdot \ba$ for some large value $B = \poly(n)$ and $y = B \cdot \bb$. Suppose that Alice and Bob want to determine the case where (1) $\Delta(\ba, \bb) = k$, and (2) $\Delta(\ba, \bb) = k + \eps k$. For the case one, we have $\norm{(\bx - \by)_{-k}}_p^p = 0$ whereas for case two we have $\norm{(\bx - \by)_{-k}}_p^p = \eps k \cdot B^p$. Hence, if there is a streaming algorithm for the trimmed $k$-sum problem with at most $\eps k |a_{k - \eps k}|^p$, then Alice and Bob can use it to design a protocol to solve the Gap Hamming problem, from which we get a $\Omega(\eps^{-2}  \log k)$ bits of space lower bound. 
\end{proof}

Note that we can scale the $\eps$ to $\eps' = \poly(\eps / \log n)$, which shows that for any $\poly(\log n / \eps)$ bits of space, the $\poly(\eps / \log n) \cdot |a_{k - \eps k}|^p$ error is best possible.

\section{Experiments} \label{sec:exper}
In this section, we evaluate the empirical performance of our algorithm on the following three datasets. All of the experiments are conducted on a laptop with a 2.42GHz CPU and 16GB RAM.

\begin{itemize}
    \item \textbf{Synthetic}: we generate a vector of size 10 million where $k$ entries are random integers between 10 and 100 thousand and the rest of the entries are between 1 and 100. This ensures that the condition we set for our top $k$ algorithm is likely met. 

    \item \textbf{AOL}\footnote{\color{blue}{https://www.kaggle.com/datasets/dineshydv/aol-user-session-collection-500}}~\cite{PCT2006picture}: we create the underlying frequency vector based on the query of each entry of the dataset. Specifically, each unique query corresponds to an entry in our underlying frequency vector, and the value of that entry is the number of times the query is in the dataset. The frequency vector had size about 1.2 million with a max frequency of 98,554.

    \item \textbf{CAIDA}\footnote{\color{blue}{https://publicdata.caida.org/datasets/topology/ark/}}~\cite{caida_ipv4_dnsnames}: we create the frequency vector based on DNS names. For the subset of the dataset we use, the frequency vector had size about 400 thousand with a max frequency of 48.
\end{itemize}

\begin{figure}[tb]
\centering
    \includegraphics[width=0.35\textwidth]{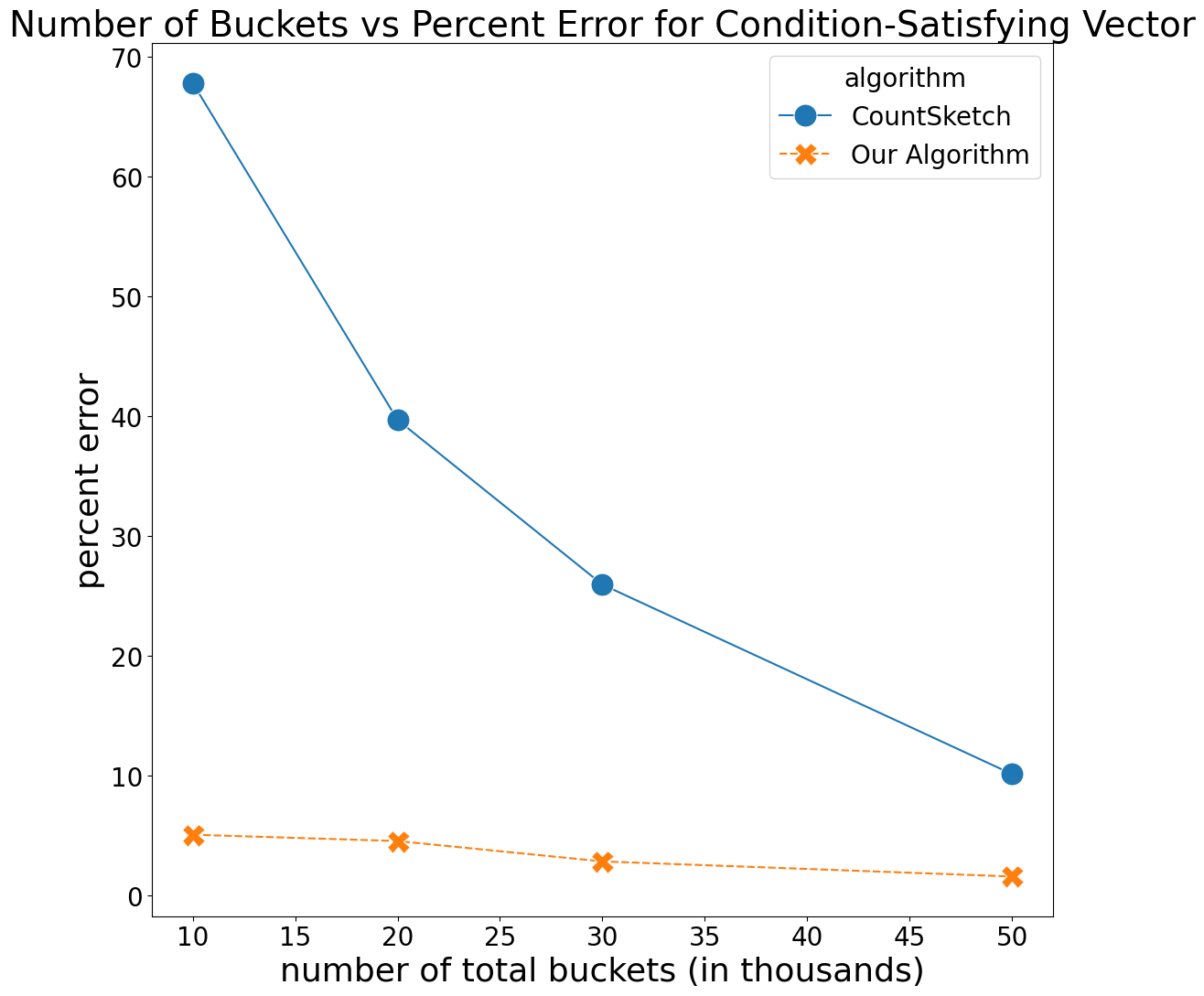}
    \includegraphics[width=0.3\textwidth]{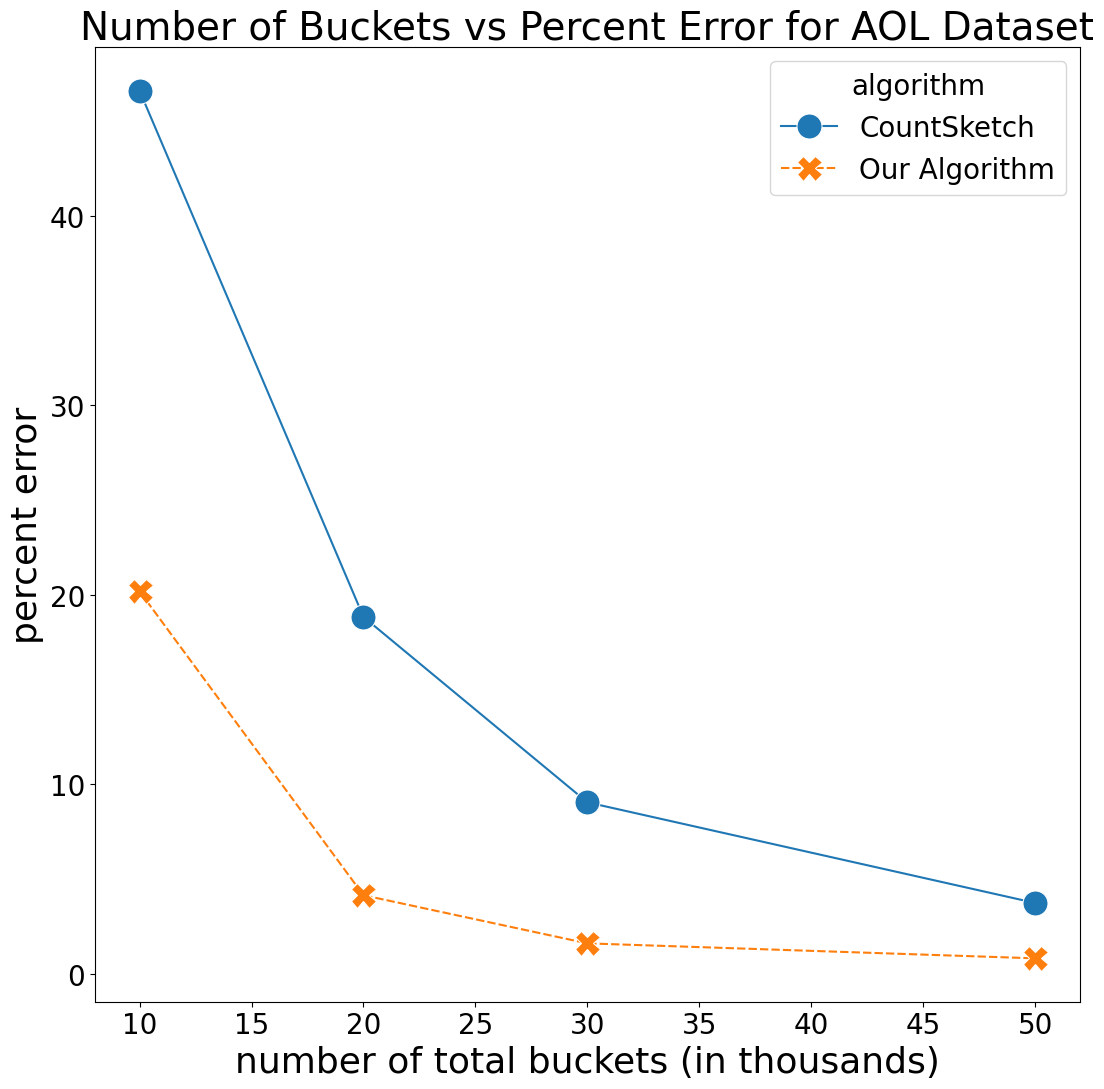}
    \includegraphics[width=0.31\textwidth]{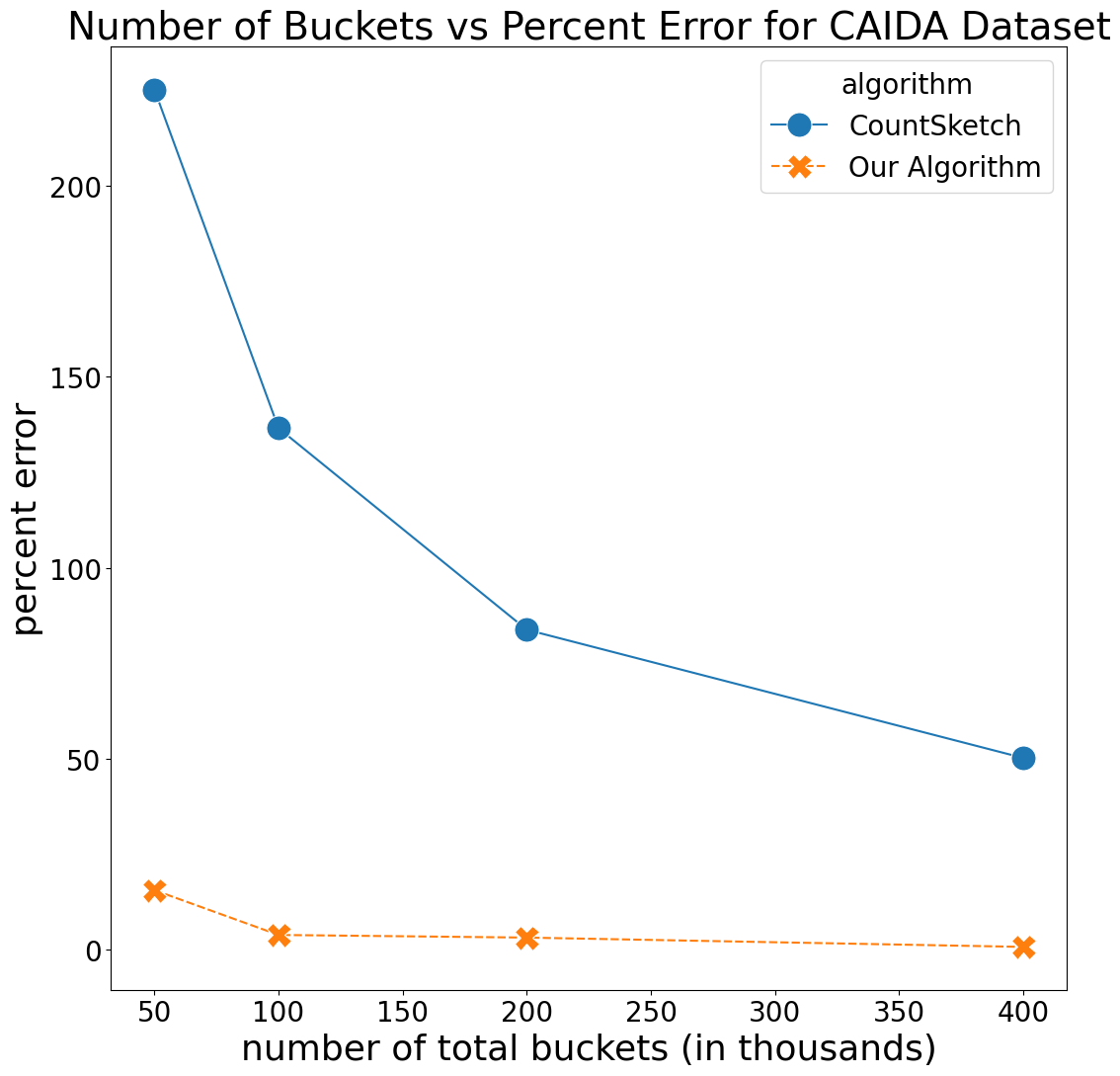}
\caption{Number of Buckets vs Error (Synthetic, AOL, CAIDA respectively).}
\label{fig}
\end{figure}

\paragraph{Experimental setting.}
In our implementations of Algorithm~\ref{alg:topk}, we make standard modifications which are done in the practical implementation of streaming algorithms. In particular, we only have a constant number of subsampling levels and level sets in our implementation. In all the three datasets, we consider the $F_1$ moment of the top $k = 1000$ frequencies for simplicity and interpretability. This can be extended to more general $F_p$ as well.

We compare our algorithm to the classical Count-Sketch across a range of total bucket sizes, using relative error with respect to the ground truth as the evaluation metric. In this context, the bucket size refers to the {\em{total}} number of buckets used in the implementation. For Count-Sketch, vector items are hashed into buckets and recovered across multiple independent repetitions. The final estimate for each item is obtained by taking the median of these repetitions. The total number of buckets is therefore the number of buckets per repetition multiplied by the number of repetitions (i.e., the number of medians taken). We vary the number of repetitions from 1 to 10 and report the lowest error achieved, as there is a trade-off between the number of repetitions and the number of buckets per repetition.

Our algorithm incorporates multiple subsampling levels, each containing a smaller Count-Sketch structure with several repetitions. As a result, the total number of buckets is given by the product of the number of subsampling levels, the number of repetitions, and the number of buckets in each Count-Sketch structure.

\paragraph{Results.} The comparison result is presented in \Cref{fig}, which suggests that our algorithm consistently outperforms the classical Count-Sketch across a range of total bucket sizes.

For the AOL dataset, as shown in \Cref{fig}, both Count-Sketch and our algorithm exhibit decreasing error as the total number of buckets increases, which aligns with expectations. However, our algorithm consistently outperforms Count-Sketch across all tested bucket sizes. Specifically, for bucket sizes of 10,000, 20,000, 30,000, and 50,000, Count-Sketch yields relative errors of 46.59\%, 18.83\%, 9.05\%, and 3.74\%, respectively. In contrast, our algorithm achieves significantly lower errors of 20.18\%, 4.15\%, 1.61\%, and 0.82\%. These results indicate that Count-Sketch requires substantially more space to match the accuracy of our method. This performance gap is expected, as Count-Sketch has a linear dependence on $k$ to achieve a provable guarantee, whereas our algorithm does not. For the CAIDA dataset, the overall trends are similar to those observed for the AOL dataset. However, we note that both algorithms required a larger number of total buckets to achieve reasonable error rates. This is likely due to the underlying frequency vector being flatter, with less distinction between the top-$k$ entries and the remainder. Despite this increased difficulty, our algorithm continues to outperform Count-Sketch across all tested bucket sizes. Specifically, for bucket sizes of 50,000, 100,000, 200,000, and 400,000, Count-Sketch yields relative errors of 225.14\%, 136.50\%, 83.75\%, and 50.18\%, respectively, while our algorithm achieves substantially lower errors of 15.57\%, 3.85\%, 3.16\%, and 0.71\%.

For the synthetic dataset, our algorithm again outperforms Count-Sketch and achieves comparable or better accuracy with significantly less space. For bucket sizes of 10,000, 20,000, 30,000, and 50,000, Count-Sketch yields errors of 67.81\%, 39.68\%, 25.93\%, and 10.11\%, respectively, while our algorithm achieves much lower errors of 5.05\%, 4.52\%, 2.82\%, and 1.56\%. Notably, the number of buckets required to obtain low error is relatively small compared to the size of the underlying frequency vector. This aligns with our expectations, as the synthetic dataset was constructed such that the top-$k$ entries are significantly larger than the rest, making them easier to identify accurately.

\begin{acks}
The authors were supported in part by Office of Naval Research award number N000142112647 and a Simons Investigator Award. Honghao Lin was supported in part by a CMU Paul and James Wang Sercomm Presidential Graduate Fellowship. Hoai-An Nguyen was supported in part by NSF GRFP award number DGE2140739. 
\end{acks}

\bibliographystyle{ACM-Reference-Format}
\bibliography{sample-base}


\end{document}